\documentclass[a4paper,11pt]{article}

\usepackage[utf8]{inputenc}
\usepackage{amsthm,amsmath,amssymb,amsfonts,tikz}
\usepackage{color}
\usepackage[pdftex, bookmarksnumbered=true]{hyperref}

\usepackage{titlesec}
\titleformat*{\section}{\Large \normalfont \bfseries}
\titleformat*{\subsection}{\large \normalfont \bfseries}
\usepackage{setspace}

\newcommand{\nc}{\newcommand}
\nc{\ny}{\nonumber}
\nc{\ra}{\rangle}
\nc{\la}{\langle}
\nc{\lb}{\left(}
\nc{\rb}{\right)}
\nc{\pt}{\partial}
\nc{\D}{\Delta}
\nc{\kr}{\mathcal}
\nc{\lmb}{\lambda}
\nc{\sg}{\sigma}
\nc{\bbl[1]}{\lbrack #1\rbrack{}}
\nc{\ep}{\epsilon}
\nc{\ora}{\overrightarrow}
\nc{\td}{\widetilde}
\nc{\tq}{\theta}
\nc{\seteq}{\mathbin{:=}}
\nc{\simto}{\xrightarrow{\,\sim\,}}
\nc{\Z}{{\mathbb Z}}
\newcommand{\NSR}{\textsf{NSR}}
\newcommand{\sfF}{\textsf{F}}
\nc{\Vir}{\textsf{Vir}}
\nc{\Asl}{\widehat{\mathfrak{sl}}}
\nc{\rmF} {\mathrm{F}}
\nc{\calM}{\mathcal{M}}
\nc{\calA}{\mathcal{A}}
\nc{\calF}{\mathcal{F}}
\nc{\calL}{\mathcal{L}}
\nc{\calU}{\mathcal{U}}
\nc{\calE}{\mathcal{E}}
\nc{\calV}{\mathcal{V}}
\nc{\calH}{\mathcal{H}}
\nc{\sfV}{\textsf{V}}
\nc{\tr}{\mathrm{Tr}}
\nc{\red}{\color{red}}
\nc{\NS}{\scriptscriptstyle{\textsf{NS}}}
\nc{\vac}{\varnothing}
\nc{\G}{\textsf{G}}
\nc{\pure}{\mathrm{pure}}
\numberwithin{equation}{section}
\newtheorem{thm}{Theorem}[section]
\newtheorem{prop}{Proposition}[section]
\newtheorem{Lemma}{Lemma}[section]
\newtheorem{Remark}{Remark}[section]
\title{Bilinear equations on Painlev\'e $\tau$~functions from CFT}
\date{}
\author{M.~A.~Bershtein, A.~I.~Shchechkin}

\begin{document}
\maketitle

\begin{abstract}\vspace*{2pt}
In 2012 Gamayun, Iorgov, Lisovyy conjectured an explicit expression for the Painlev\'e VI $\tau$~function in terms of the Liouville conformal blocks with central charge $c=1$. We prove that proposed expression satisfies Painlev\'e VI $\tau$~function bilinear equations (and therefore prove the conjecture).

The proof reduces to the proof of bilinear relations on conformal blocks. These relations were studied using the embedding of a direct sum of two Virasoro algebras into a sum of Majorana fermion and Super Virasoro algebra. In the framework of the AGT correspondence the bilinear equations on the conformal blocks can be interpreted in terms of instanton counting on the minimal resolution of $\mathbb{C}^2/\mathbb{Z}_2$ (similarly to Nakajima-Yoshioka blow-up equations).

\end{abstract}

\tableofcontents

\newpage

\section{Introduction}

Painlev\'e equations were introduced more than 100 years ago. Solutions of these equations (Painlev\'e transcendents) are important special functions with many applications including integrable models and random matrix theory. Maybe the most natural mathematical framework for the Painlev\'e equations is the theory of monodromy preserving deformations. The Painlev\'e VI equation is equivalent to the simplest nontrivial example namely isomonodromic deformation of rank 2 linear system on $\mathbb{CP}^1$ with four regular singular points.

In 2012 Gamayun, Iorgov, Lisovyy conjectured \cite{GIL1207} an explicit expression for the expansion near $t=0$ of the Painlev\'e VI $\tau$~function. This expression is an infinite sum of four-point $\mathbb{CP}^1$ conformal blocks for $c=1$ conformal field theory (CFT). In the next paper \cite{GIL1302} Gamayun, Iorgov, Lisovyy found the analogous expressions for the $\tau$~functions of the Painlev\'e V, III's equations in terms of certain limits of conformal blocks for $c=1$. See also \cite{ILT},\cite{Its:2014} for further developments of this conjecture.

It this paper we prove the Gamayun--Iorgov--Lisovyy conjecture. Note that completely different proof of this conjecture (together with a generalization to any number of points on $\mathbb{CP}^1$) was given in \cite{MCB}. We explain the main idea of our proof in the simplest Painlev\'e $\mathrm{III}'_3$ case. This equation has the form
\begin{equation}\label{eq:piiispr}
\frac{d^2q}{dt^2}=\frac{1}{q}\left(\frac{dq}{dt}\right)^2 -
 \frac{1}{t}\frac{dq}{dt}\,+\frac{2q^2}{t^2}-\frac{2}{t},
\end{equation}
Our proof is based on another form of Painlev\'e $\mathrm{III}'_3$ equation, namely on a bilinear equation on the Painlev\'e $\mathrm{III}'_3$ $\tau$~function (see Subsection \ref{ssec:Painleve} or the paper \cite{myCA} for the relation between different forms of Painlev\'e $\mathrm{III}'_3$ equations). It is convenient to write the bilinear equation by use of Hirota differential operators with respect to $\log t$. Then the $\tau$ form of the Painlev\'e III$'_3$ equation can be written as
\begin{equation}
D^{III}(\tau(t),\tau(t))=0,\quad \text{where} \quad D^{III}=
\frac12 D^{4}_{[\log t]}-t\frac{d}{dt}D^{2}_{[\log t]}+\frac12D^{2}_{[\log t]}+2tD^{0}_{[\log t]}, \label{eq:tau3}
\end{equation}
and $D^{k}_{[x]}$ is the $k$-th Hirota operator with respect to the variable $x$ see eq. \eqref{eq:Hirota:def}. For this Painlev\'e III$'_3$ case our main result is the following 

\begin{thm}
\label{thm:1}
The expansion of the Painlev\'e $III'_3$ $\tau$~function near $t=0$ can be written as 
\begin{equation}\label{eq:tau:P3Int}
\tau(t)=\sum_{n\in\mathbb{Z}}s^n C(\sg+n) \kr{F}((\sg+n)^2|t),
\end{equation}
\end{thm}

In this theorem $\kr{F}(\sg^2|t)=\kr{F}_1(\sg^2|t)$ denotes the irregular limit of the Virasoro ($\Vir$) conformal block for the central charge $c=1$. This function is defined in terms of  representation theory of the Virasoro algebra (see Subsection \ref{ssec:chain}). The coefficients $C(\sg)$ are defined by
$ C(\sg)=1/\Bigl(\G(1-2\sg)\G(1+2\sg)\Bigr),$
where $\G(z)$ is the Barnes $\G$ function. The parameters $s$ and $\sg$ in \eqref{eq:tau:P3Int} are integration constants of  equation~\eqref{eq:piiispr}. Note that Theorem \ref{thm:1} means that the formula \eqref{eq:tau:P3Int} specifies 2 parametric set of $\tau$~function but there could be $\tau$~function which are not given by \eqref{eq:tau:P3Int}.

We substitute \eqref{eq:tau:P3Int} to \eqref{eq:tau3} 
and collect terms with the same powers of $s$. A vanishing condition of a $s^m$ coefficient reads
\begin{align}
\sum_{n\in\mathbb{Z}}\left(C(\sg+n+m)C(\sg-n) D^{III}\Bigl(\kr{F}((\sigma+n+m)^2|t),\kr{F}((\sigma-n)^2|t\Bigr)\right)=0, \label{eq:bilinP3}
\end{align}
Each summand looks similar to a conformal block for the sum of two Virasoro algebras $\Vir\oplus \Vir$. Therefore it is natural to expect that the whole sum is a conformal block of an extension of $\Vir\oplus \Vir$. And we prove that the required extension is $\Vir\oplus \Vir \subset \mathsf{F}\oplus\NSR$, where $\mathsf{F}$ is the Majorana fermion algebra and $\NSR$ is the Neveu--Schwarz--Ramond algebra, $\mathcal{N}=1$ superanalogue of the Virasoro algebra (to be more precise $\Vir\oplus \Vir \subset \mathsf{F}\oplus\NSR$ is an extension on vertex operator algebras).

The idea to use this extension comes from geometry. Recall that the AGT correspondence \cite{AGT} for the Virasoro algebra states that the conformal block in the irregular limit coincides with the Nekrasov partition function for pure $\mathcal{N}=2$ supersymmetric $U(2)$ gauge theory on $\mathbb{C}^2$. Therefore  equation \eqref{eq:bilinP3} is equivalent to a bilinear equation on the Nekrasov partition function. 

These equations resemble Nakajima-Yoshioka bilinear equations \cite{Nakajima}, which relate the Nekrasov partition function on the blow-up of $\mathbb{C}^2$ with the Nekrasov partition function on $\mathbb{C}^2$. Due to the AGT correspondence the Nakajima-Yoshioka bilinear equations are equivalent to equations which relate conformal blocks of two theories with central charges $c^{\scriptscriptstyle{(1)}}_{NY}$ and $c^{\scriptscriptstyle{(2)}}_{NY}$. The CFT meaning of this relation was explained in \cite{BFL} and the explanation is based on the embedding $\Vir\oplus \Vir \subset \Vir \times \mathcal{U}$, where $\mathcal{U}$ is a certain rational vertex operator algebra. 

We use an analogue of the Nakajima-Yoshioka equations. Geometrically this analogue corresponds to the instanton counting on the minimal resolution of $\mathbb{C}^2/\mathbb{Z}_2$. The CFT interpretation of this instanton counting was given in \cite{BBFLT} and it was based on the above mentioned embedding $\Vir\oplus \Vir \subset \mathsf{F}\oplus\NSR$. It is convenient to write the central charges of $\NSR$ and $\Vir$ algebras in terms of a parameter $b\in \mathbb{C}$. To be precise the central charges of the two Virasoro algebras are equal to
\[
c^{\scriptscriptstyle{(1)}}=1+6\dfrac{(b+b^{-1})^2}{2b(b^{-1}-b)},
\quad c^{\scriptscriptstyle{(2)}}=1+6\dfrac{(b+b^{-1})^2}{2b^{-1}(b-b^{-1})}.
\]
Therefore in the $b+b^{-1}=0$ case we get $c^{\scriptscriptstyle{(1)}}=c^{\scriptscriptstyle{(2)}}=1$, just as we want for equation \eqref{eq:bilinP3}.

It was proved in \cite{BBFLT} that the $\sfF \oplus \NSR$ Verma module $\pi^{\D^{\NS}}_{\sfF\oplus \NSR}$  is decomposed into a direct sum of $\Vir\oplus \Vir$ Verma modules
\[
\bigoplus\limits_{2n\in\mathbb{Z}}\pi^n_{\Vir\oplus \Vir}\cong \pi^{\D^{\NS}}_{\sfF\oplus \NSR}. 
\]
Therefore we get similar relation for the conformal blocks
\[
\sum\limits_{2n\in\mathbb{Z}}\Bigl( l_n^2(P,b)\, \kr{F}_{c^{\scriptscriptstyle{(1)}}}(\Delta_n^{\scriptscriptstyle{(1)}}|\beta^{\scriptscriptstyle{(1)}}q),\kr{F}_{c^{\scriptscriptstyle{(2)}}}(\Delta_{n}^{\scriptscriptstyle{(2)}}|\beta^{\scriptscriptstyle{(2)}}q)\Bigr)=\kr{F}_{c_{\NS}},
\]
where $\mathcal{F}_{c^{\scriptscriptstyle{(1)}}}$ and $\mathcal{F}_{c^{\scriptscriptstyle{(2)}}}$ denote conformal blocks for the first and second Virasoro in $\Vir\oplus \Vir$, $\kr{F}_{c^{\NS}}$ denotes an $\NSR$ conformal block. The coefficients $l_n(P,b)$ are called blow-up factors in \cite{BBFLT} due to their geometric origin. In the language of conformal field theory these coefficients are closely related to the structure constants of the theory with $\Vir\oplus \Vir$ symmetry. 

In Section \ref{sec:bilin} we consider more general conformal blocks (following \cite{BFL}) and prove the relation 
\begin{equation}
\sum\limits_{2n\in\mathbb{Z}}l_n^2(P,b)\, D^{III}_b\Bigl(\kr{F}_{c^{\scriptscriptstyle{(1)}}}(\Delta_n^{\scriptscriptstyle{(1)}}|\beta^{\scriptscriptstyle{(1)}}q),\kr{F}_{c^{\scriptscriptstyle{(2)}}}(\Delta_{n}^{\scriptscriptstyle{(2)}}|\beta^{\scriptscriptstyle{(2)}}q)\Bigr)=0 \label{eq:bilinP3def}
\end{equation}
where the operator $D^{III}_b$ is written in terms of $b$-deformed Hirota differential operators. If we set $b=i$ we get the relation \eqref{eq:bilinP3} and therefore prove Gamayun--Iorgov--Lisovyy conjecture.

More precisely for the relation \eqref{eq:bilinP3} with $m=0$ we need to show that the coefficients $l_n^2(2i\sigma,i)$ are proportional to $\frac{C(\sg+n)C(\sg-n)}{C(\sigma)^2}$ as functions on $n$.  The expressions for $l_n(P,b)$  were given in \cite{BBFLT} without proof. They were  computed in the recent paper \cite{HJ} by use of Dotsenko-Fateev type integrals.

We calculate $l_n(P,b)$ in Section \ref{sec:conf:block} using a completely different approach. Namely we imitate the computation of structure constants in the Liouville CFT \cite{ZZbook},\cite{JTL}, based on the associativity of OPE and properties of degenerate fields $\phi_{1,2}$. It is interesting to note that in \cite{ZZbook},\cite{JTL} a monodromy of correlation functions is trivial due to the coupling of the chiral and antichiral CFT, in contrast to our case where we have only chiral CFT and monodromy is trivial due to the relation between the central charges $c^{\scriptscriptstyle{(1)}}$ and $c^{\scriptscriptstyle{(2)}}$. 

As was explained above the calculation of $l_n(P,b)$ and the proof of relation \eqref{eq:bilinP3def} leads to a proof of Theorem~\ref{thm:1}. The Painlev\'e VI case is studied in a similar way. 
 
The paper is organized as follows. In Section \ref{sec:Algebras} we recall the main properties of the embedding of algebras $\Vir\oplus \Vir \subset \mathsf{F}\oplus \NSR$. In Subsection \ref{ssec:vac} we prove Theorem \ref{thm:Vac} which describes the  $\mathsf{F}\oplus \NSR$ vacuum module as a module over the vertex operator subalgebra  $\Vir\oplus \Vir$. We do not use this theorem in the rest of the paper so the reader can safely skip it. 

Section \ref{sec:conf:block} is devoted to conformal blocks, in Subsection \ref{ssec:chdec} we recall the relation between $\Vir\oplus \Vir$ and $\mathsf{F}\oplus \NSR$ vertex operators, and in Subsection \ref{ssec:melem} we calculate the blow-up factors $l_n(P,b)$.

Section \ref{sec:bilin} is devoted to bilinear relations. In Subsection \ref{ssec:Painleve} we recall the necessary background on the Painlev\'e equations and the isomonodromic $\tau$~functions, in Subsections \ref{ssec:PIII} and \ref{ssec:P6case} we prove the Painlev\'e III and Painlev\'e VI $\tau$~functions conjectures. In Subsection \ref{ssec:calc} we show that the bilinear relations on conformal blocks provide an efficient method for the calculation of the conformal block expansion. 

In Section \ref{sec:geom} we discuss the AGT meaning of our results. In particular we recall the arguments of~\cite{BBFLT}, which reduces the proof of AGT relation for the $\NSR$ algebra to the calculation of blow-up factors $l_n(P,b)$ calculated in Subsection \ref{ssec:melem} (and also in \cite{HJ}).  

Finally in Section \ref{sec:disc} we formulate some questions to study.

\section{Algebras and representations} \label{sec:Algebras}

\subsection{Verma modules} \label{ssec:FplNSR}


The Virasoro algebra (which we denote by $\Vir$) is generated by  $L_n$, $n\in\mathbb{Z}$ with relations
\[
[L_n,L_m]=(n-m)L_{n+m}+\frac{n^3-n}{12}c\delta_{n+m,0}. 
\]
Here $c$ is an additional central generator, which acts on representations below as multiplication by a complex number. Therefore we consider $c$ as a complex number, which we call central charge.

Denote the Verma module of $\Vir$ by $\pi^{\D}_{\Vir}$. This module is generated by a highest weight vector $|\D\ra$
\[
L_0|\D\ra=\D|\D\ra, \quad L_n|\D\ra=0, \quad n>0,
\]
where $\D \in \mathbb{C}$ is called the weight of $|\D\ra$.
The representation space is spanned by vectors, obtained by the action of the operators $L_{-n}, \,n>0$ on $|\D\ra$.

The $\mathsf{F}\oplus \NSR$ algebra is a direct sum of the free-fermion algebra $\mathsf{F}$ with generators $f_r$ ($r\in\mathbb{Z}+\frac12$) and $\NSR$ (Neveu-Schwarz-Ramond or Super Virasoro) algebra with generators $L_n, G_r$ ($n\in\mathbb{Z}, r\in\mathbb{Z}+\frac12$). These generators satisfy commutation relations
\begin{equation}
\begin{aligned}
\{f_r,f_s\}&=\delta_{r+s,0}, \\ \{f_r,G_s\}&=[f_r,L_n]=0 \\
[L_n,L_m]&=(n-m)L_{n+m}+ \frac{(n^3-n)}8c_{\NS}\delta_{n+m,0}\\
\{G_r,G_s\}&=2L_{r+s}+\frac{1}{2}c_{\NS}\left(r^2-\frac{1}{4}\right)\delta_{r+s,0}\\~
[L_n,G_r]&=\left(\frac{1}{2}n-r\right)G_{n+r}.
\end{aligned}\label{eq:FNSRrelations}
\end{equation}
It is convenient to express the central charge by
\[
c_{\NS}=1+2Q^2,\quad \text{where $Q=b^{-1}+b$}
\]
Here and below we choose the indices $r$ of $G_r$ and $f_r$ to be half-integer, i.e. we work in the so-called $\textsf{NS}$ sector of our algebras.

We denote by $\pi^{\D^{\NS}}_{\sfF\oplus \NSR}$ a Verma module of the $\sfF\oplus \NSR$ algebra. This module is isomorphic to a tensor product of Verma modules $\pi_\sfF$ and $\pi^{\D^{\NS}}_\NSR$ which are generated by the highest weight vectors 
$|1\ra$ and $|\D^{\NS}\ra$ correspondingly defined by
\[
f_r|1\ra=0, \quad r>0,
\]
and
\[
L_0|\D^{\NS}\ra=\D^{\NS}|\D^{\NS}\ra, \qquad L_n|\D^{\NS}\ra=0, \quad G_r|\D^{\NS}\ra=0, \quad n,r>0.
\]
The representation space is spanned by vectors obtained by the action of  generators with negative indices on the highest weight vector. We denote the highest weight vector $|1\ra\otimes |\D^{\NS}\ra$ as $\overline{|\D^{\NS}\ra}$

Recall the free-field realization of the $\NSR$ algebra. Consider the algebra generated by $c_n, n\in\mathbb{Z}$ and $\psi_r, r\in\mathbb{Z}+\frac12$ (free boson and free fermion) with relations
\[
[c_n,c_m]=n\delta_{n+m,0}, \quad
[c_n,\psi_r]=0, \quad
\{\psi_r,\psi_s\}=\delta_{r+s,0}.
\]
We denote by $\hat{P}$ the operator $c_0$. Then a Fock representation of this algebra is generated by a vacuum vector $|P\ra$ such that $\psi_r|P\ra=c_n|P\ra=0$, $\hat{P}|P\ra=P|P\ra$ for $r,n >0$. On this Fock module we can define an action of the  $\NSR$ algebra by formulae
\begin{equation}
\begin{aligned}
L_n=&\frac12\sum_{k\neq 0,n}c_k c_{n-k}+\frac12\sum_r\left(r-\frac{n}{2}\right)\psi_{n-r}\psi_r+\frac{i}{2}\left(Qn\mp 2\hat{P}\right)c_n,\quad n\neq 0,\\
L_0=&\sum_{k>0} c_{-k} c_k+\sum_{r>0}r \psi_{-r}\psi_r+\frac{1}{2} \left(\frac{Q^2}{4}-\hat{P}^2 \right),\label{eq.FFR} \\
G_r=&\sum_{n\neq 0}c_n \psi_{r-n}+i (Qr\mp \hat{P})\psi_r.
\end{aligned}
\end{equation}
We say that $P$ is generic if $P\not \in \{\frac12(mb^{-1}+nb)|m,n \in \mathbb{Z}, mn>0\}$. For generic $P$ the $\NSR$ module defined by \eqref{eq.FFR} is irreducible and isomorphic to the Verma module $\pi_{\NSR}^{\Delta^{\NS}}$, where
\[
\D^{\NS}=\D^{\NS}(P,b)=\frac12 \left(\frac{Q^2}{4}-P^2\right), \qquad |P\ra=|\D^{\NS}\ra 
\]
The sign $\mp$ in \eqref{eq.FFR} refers to the existence of two free-field representations of the same Verma module. We denote the corresponding operators by $c_n^+, \psi_r^+$ and $c_n^-, \psi_r^-$. These operators are conjugated by some unitary operator acting on $\pi_{\NSR}^{\Delta^{\NS}}$ (the so-called super Liouville reflection operator).

As the main tool we shall use the $\Vir\oplus \Vir$ subalgebra in  the $\sfF\oplus \NSR$ algebra (following \cite{CSS}, \cite{Lashkevich}). The embedding of the $\Vir\oplus \Vir$  algebra is defined by formulae
\begin{equation}
\begin{aligned}
&L_{n}^{\scriptscriptstyle{(1)}}=\frac{b^{-1}}{b^{-1}-b}L_{n}-\frac{b^{-1}+2b}{2(b^{-1}-b)}\sum\limits_{r\in\mathbb{Z}-1/2}r:f_{n-r}f_{r}:+\frac{1}{b^{-1}-b}\sum\limits_{r\in\mathbb{Z}-1/2}f_{n-r}G_r\\
&L_{n}^{\scriptscriptstyle{(2)}}=\frac{b}{b-b^{-1}}L_{n}-\frac{b+2b^{-1}}{2(b-b^{-1})}\sum\limits_{r\in\mathbb{Z}-1/2}r:f_{n-r}f_{r}:+\frac{1}{b-b^{-1}}\sum\limits_{r\in\mathbb{Z}-1/2}f_{n-r}G_r \label{Vir12}
\end{aligned}
\end{equation}
Note that the expressions for $L_{n}^{(\eta)}$, $\eta=1,2$ contain infinite sums and belong to certain completion of the universal enveloping algebra of $\sfF\oplus \NSR$. The operators $L_{n}^{(\eta)}$ act on any highest weight representation of $\sfF\oplus \NSR$. One can say that $\Vir \oplus \Vir$ is a vertex operator subalgebra of $\sfF\oplus \NSR$  (see the next subsection).

It is convenient to express the central charge of the Virasoro algebra and the highest weights of the Verma module by
\begin{align}
\Delta(P,b)=\frac{Q^2}{4}-P^2, \;\; c(b)=1+6Q^2, \quad \text{where} \quad Q=b+b^{-1} \label{eq:cDelta:Vir}
\end{align} 
Then the central charges of these $\Vir^{\scriptscriptstyle{(1)}}$ and $\Vir^{\scriptscriptstyle{(2)}}$ subalgebras are equal to 
\begin{align}
c^{(\eta)}=c(b^{(\eta)}),\; \eta=1,2,\quad \text{where}\quad  (b^{\scriptscriptstyle{(1)}})^2=\frac{2 b^2}{1-b^2}, \quad (b^{\scriptscriptstyle{(2)}})^{-2}=\frac{2 b^{-2}}{1-b^{-2}}. \label{cc}
\end{align}
Note that the symmetry $b\leftrightarrow b^{-1}$ permutes $\Vir^{\scriptscriptstyle{(1)}}$ and $\Vir^{\scriptscriptstyle{(2)}}$. Here and below $b^2\neq 0,1$.

Now consider the space $\pi^{\D^{\NS}}_{\sfF\oplus \NSR}$ as a module over $\Vir \oplus \Vir$. Clearly, the vector $\overline{|\Delta^{\NS}\ra}=|1\ra\otimes |\Delta^{\NS}\ra$ is a highest weight vector with respect to $\Vir \oplus \Vir$. This vector generates a Verma module $\pi^{\D^{\scriptscriptstyle{(1)}},\D^{\scriptscriptstyle{(2)}}}_{\Vir\oplus \Vir}$. The highest weight $(\D^{\scriptscriptstyle{(1)}},\D^{\scriptscriptstyle{(2)}})$ can be found from \eqref{Vir12}, namely
\begin{align}\label{eq:DeltaVirNSR}
\D^{\scriptscriptstyle{(1)}}=\frac{b^{-1}}{b^{-1}-b}\D^{\NS}, \quad  \D^{\scriptscriptstyle{(2)}}=\frac{b}{b-b^{-1}}\D^{\NS}
\end{align}
But the whole space $\pi^{\D^{\NS}}_{\sfF\oplus \NSR}$ is larger than $\pi^{\D^{\scriptscriptstyle{(1)}},\D^{\scriptscriptstyle{(2)}}}_{\Vir\oplus \Vir}$. The following decomposition was proved in \cite{BBFLT}.

\begin{prop} \label{pr:FNSR}
 For generic $P$ the space $\pi^{\D^{\NS}}_{\sfF\oplus \NSR}$ is isomorphic to the sum of $\Vir\oplus \Vir$ modules
\begin{align}
\pi^{\D^{\NS}}_{\sfF\oplus \NSR}\cong\bigoplus\limits_{2n\in\mathbb{Z}}\pi^n_{\Vir\oplus \Vir} \label{FxNSR}.
\end{align}
The highest weight $(\D^{\scriptscriptstyle{(1)}}_{n},\D^{\scriptscriptstyle{(2)}}_{n})$ of the Verma module $\pi^n_{\Vir\oplus \Vir}$ is defined by $\Delta^{(\eta)}_{n}=\Delta(P^{(\eta)}_n,b^{(\eta)}),$ $\eta=1,2$, 
where 
\begin{align} 
P^{\scriptscriptstyle{(1)}}_{n}=P^{\scriptscriptstyle{(1)}}+nb^{\scriptscriptstyle{(1)}}, \; P^{\scriptscriptstyle{(2)}}_{n}=P^{\scriptscriptstyle{(2)}}+n\left(b^{\scriptscriptstyle{(2)}}\right)^{-1}, \qquad P^{\scriptscriptstyle{(1)}}=\frac{P}{\sqrt{2-2 b^2}}, \; P^{\scriptscriptstyle{(2)}}=\frac{P}{\sqrt{2-2 b^{-2}}}\label{momentum}.
\end{align}
\end{prop}

\begin{figure}[h]
\begin{center}
\begin{tikzpicture}[x=4em,y=4em]
\draw[dashed] (-3,0.5) -- (3,0.5);
\draw (-2,0.45) -- (-2,0.5) node[anchor=south] {$-1$};
\draw (-1,0.45) -- (-1,0.5) node[anchor=south] {$-0.5$};
\draw (0,0.45) -- (-0,0.5) node[anchor=south] {$0$};
\draw (1,0.45) -- (1,0.5) node[anchor=south] {$0.5$};
\draw (2,0.45) -- (2,0.55) node[anchor=south] {$1$};
\draw (3.2,0.5) node[anchor=west] {$n$};
\draw[dashed] (-3.3,-3.2) -- (-3.3,0.2);
\draw (-3.3,-3.3) node[anchor=north] {$L_0+L_0^f$};
\draw (-3.25,0) -- (-3.35,0) node[anchor=east] {$\Delta^{\NS}$};
\draw (-3.25,-0.5) -- (-3.35,-0.5) node[anchor=east] {$\Delta^{\NS}+0.5$};
\draw (-3.25,-1) -- (-3.35,-1) node[anchor=east] {$\Delta^{\NS}+1$};
\draw (-3.25,-1.5) -- (-3.35,-1.5) node[anchor=east] {$\Delta^{\NS}+1.5$};
\draw (-3.25,-2) -- (-3.35,-2) node[anchor=east] {$\Delta^{\NS}+2$};
\draw (-3.25,-2.5) -- (-3.35,-2.5) node[anchor=east] {$\Delta^{\NS}+2.5$};
\draw (1.5,-3.5) -- (2,-2) -- (2.5,-3.5);
\draw (1.4,-2.1) -- (1,-0.5) -- (0.6,-2.1);
\draw (-0.7,-2.7) -- (0,0) -- (0.7,-2.7);
\draw (-1.4,-2.1) -- (-1,-0.5) -- (-0.6,-2.1);
\draw (-1.5,-3.5) -- (-2,-2) -- (-2.5,-3.5);
\end{tikzpicture}
\caption{\leftskip=1in \rightskip=1in Decomposition of $\pi^{\D^{\NS}}_{\sfF\oplus \NSR}$ into direct sum of representations of the algebra $\mathsf{Vir}\oplus\mathsf{Vir}$. Each interior angle corresponds to the Verma module~$\pi^n_{\Vir\oplus \Vir}$.}
	\label{NS-decompos-pic}
\end{center}
\end{figure}

	
The proof is based on an explicit construction of the highest weight vectors of the representations $\pi^n_{\Vir\oplus \Vir}$. We will use operators of free field realization $\psi_r^{\pm}$, which anticommute with $f_s$. Let us introduce another fermion operators acting on $\pi^{\D^{\NS}}_{\sfF\oplus \NSR}$
\[
\chi_r^{\mp}=f_r-i\psi^{\mp}_r
\]
Then it can be checked that the vectors $|P,n\ra$, $2n \in \mathbb{Z}$ defined by the formulae
\begin{equation}
\begin{aligned}
|P,n\ra=\Omega_n(P)\prod\limits_{r=1/2}^{(4n-1)/2}\chi^-_{-r}\overline{|\Delta^{\NS}\ra},\; n>0,  \qquad |P,n\ra=\Omega_n(P)\prod\limits_{r=1/2}^{(-4n-1)/2}\chi^+_{-r}\overline{|\Delta^{\NS}\ra},\; n<0,  \\  |P,0\ra=\overline{|P\ra}=\overline{|\Delta^{\NS}\ra}. 
\end{aligned}\label{eq:|P,n}
\end{equation}
satisfy highest weight vector equations
\[
L_0^{(\eta)}|P,n\ra=\D_n^{(\eta)}|P,n\ra, \quad L_k^{(\eta)}|P,n\ra=0, \quad k>0,\; 2n\in \mathbb{Z},\;\eta=1,2.
\]
Since $P \longleftrightarrow -P$ symmetry permutes $\chi_r^+$ and $\chi_r^-$ we have $|P,n\ra\equiv |-P,-n\ra$. Remark that here $\Omega_n(P)$ are arbitrary normalization constants. We shall specify them below.

We can write that $|P,n\ra=|\D^{\scriptscriptstyle{(1)}}_{n}\ra\otimes|\D^{\scriptscriptstyle{(2)}}_{n}\ra$.
Note that the highest weights of $|P,n\ra$ satisfy the relation
\[
\D^{\scriptscriptstyle{(1)}}_{n}+\D^{\scriptscriptstyle{(2)}}_{n}=\D^{\NS}+2n^2. 
\]
This relation follows from \eqref{momentum}.
For generic $P$ the vectors $|P,n\ra$ are linear independent and generate the Verma modules over the algebra $\Vir \oplus \Vir$.

The isomorphism from Proposition \ref{pr:FNSR} follows from the coincidence of the characters of the l.h.s. and the r.h.s. of \eqref{FxNSR}.
The character of the module $V$ is $\mathrm{ch}(V)=\left.\mathrm{Tr}\right|_V q^{L_0}$. The characters of Verma modules equal
\[
\mathrm{ch}(\pi_\Vir^\Delta)=q^\Delta\prod_{k=1}^\infty \frac{1}{1-q^k}, \qquad \mathrm{ch}(\pi_{\sfF\oplus\NSR}^{\Delta^{\NS}})=q^{\Delta^{\NS}}\prod_{k=1}^\infty \frac{(1+q^{k-\frac12})^2}{1-q^k}.
\]
Using the Jacobi triple product identity
\begin{equation}
\prod_{k=1}^{\infty}(1-q^{2k})(1+q^{2k-1}y^2)(1+q^{2k-1}y^{-2})=\sum_{k=-\infty}^{\infty}q^{k^2}y^{2k} \label{JTP}
\end{equation}
in the case $y \mapsto 1, q \mapsto q^{1/2}$ we have the necessary equality of characters 
\[
\mathrm{ch}(\pi_{\sfF\oplus \NSR}^{\D^{\NS}})= q^{\Delta^{\NS}}\prod_{k=1}^\infty\frac{(1+q^{k-\frac12})^2}{(1-q^k)}=
\sum_{2n\in\mathbb{Z}}q^{\Delta^{\NS}+2n^2}\prod_{k=1}^\infty\frac{1}{(1-q^k)^2}=\sum_{2n\in\mathbb{Z}}\mathrm{ch}(\pi^n_{\Vir\oplus \Vir})
\]

\subsection{Vacuum module}
\label{ssec:vac}
In this subsection we revisit the relation between the $\sfF\oplus\NSR$ and the $\Vir\oplus \Vir$ algebras for generic central charges (in terms of the parameter $b$ it means that $b^2 \not \in \mathbb{Q}$). Some formulas for unitary minimal model case ($b^2=-(n+2)/n$) were given in \cite{CSS}. We do not use the results of this subsection in the rest of the paper.

 We use the language of vertex operator algebras, (VOA for short), see e.g. \cite{FrenkelBook}. Recall that a vector space $V$ is called a vacuum module of VOA if any vector $v\in V$ corresponds to a \emph{current} i.e. a power series of operators $Y(v;q)=\sum Y_n q^{-n}$, where $Y_n\in \mathrm{End}(V)$. This correspondence $v \longleftrightarrow Y(v;z)$ is called the \emph{operator-state correspondence}. In the definition of the vertex operator algebra the correspondence $v \longleftrightarrow Y(v;q)$ should satisfy certain conditions, namely vacuum axiom, translation axiom and locality axiom.

The vacuum module $\mathrm{Vac}$ for the $\sfF \oplus \NSR$ algebra is generated by the vector $|\vac\ra$ defined by
\[ f_r |\vac\ra=0,\; \text{for $r\geq \frac12$;}\qquad G_r |\vac\ra=0,\; \text{for $r\geq -\frac12$}; \qquad L_n |\vac\ra=0\, \text{for $n\geq -1$}.\]
The simplest examples of the operator-state correspondence are
\begin{align*}
f_{-1/2}|\vac\ra \longleftrightarrow f(q)=\sum\nolimits_{r \in \mathbb{Z}+\frac12} &f_r q^{-r-1/2},\qquad
G_{-3/2}|\vac\ra \longleftrightarrow G(q)=\sum\nolimits_{r \in \mathbb{Z}+\frac12} G_r q^{-r-3/2},\\
&L_{-2}|\vac\ra \longleftrightarrow T(q)=\sum\nolimits_{n \in \mathbb{Z}} L_n q^{-n-2}.
\end{align*}
Other currents in $\sfF\oplus \NSR$ vertex operator algebra can be obtained by use of normal ordered products from $f(q)$, $G(q)$, $T(q)$ and its derivatives (see \cite[Th. 4.4.1]{FrenkelBook}). The current $T(q)$ is called the \emph{stress-energy tensor}.

Formulae \eqref{Vir12} defines two currents $T^{\scriptscriptstyle{(1)}}(q)$ and $T^{\scriptscriptstyle{(2)}}(q)$ i.e. define the vertex operator subalgebra $\Vir \oplus \Vir$ in the vertex operator algebra $\sfF \oplus \NSR$. We consider $\sfF \oplus \NSR$ as an extension of $\Vir \oplus \Vir$.

\begin{Lemma} 
The vertex operator algebra $\sfF \oplus \NSR$ is generated by currents $T^{\scriptscriptstyle{(1)}}(q)$, $T^{\scriptscriptstyle{(2)}}(q)$ and $f(q)$ 
\end{Lemma}
\begin{proof}
It is enough to express currents $G(q)$ and $T(q)$. It follows from \eqref{Vir12} that
\begin{align*}
T(q)&=T^{\scriptscriptstyle{(1)}}(q)+T^{\scriptscriptstyle{(2)}}(q)-\frac12:f'(q)f(q):\\
G(q)&=\frac{b+2b^{-1}}{2\pi i}\oint_q dz T^{\scriptscriptstyle{(1)}}(q)f(z)+\frac{b^{-1}+2b}{2\pi i}\oint_q dz T^{\scriptscriptstyle{(2)}}(q) f(z)
\end{align*}
Here we used that $:\!\!f(q)f(z)\!\!:=\frac{1}{q-z}+\mathrm{reg}$ and $T(q)f(z)=\mathrm{reg}$.
\end{proof}
Note that the current $T_f(q)=\frac12:f'(q)f(q):$ is the standard fermion stress-energy tensor.

We want to describe the structure of $\mathrm{Vac}$ as a module over the $\Vir \oplus \Vir$ algebra. Recall 
that for $\Delta=\Delta_{m,n}(b)=((b^{-1}+b)^2-(mb^{-1}+nb)^2)/4$, $m,n \in \mathbb{N}$ the Verma module $\pi_\Vir^\Delta$ contains the singular vector of the level $mn$ (see e.g. \cite{Feigin Fuchs 1990}). Denote by $\mathbb{L}_{m,n}^b$ an irreducible quotient of the $\pi_\Vir^{\Delta_{m,n}}$. The superscript $b$ stresses the dependence on the central charge. 

\begin{Lemma} \label{le:f-12}
The vector $f_{-1/2}|\vac\ra$ is the highest weight vector of the $\Vir\oplus \Vir$ module $\mathbb{L}_{1,2}^{b^{\scriptscriptstyle{(1)}}}\otimes \mathbb{L}_{2,1}^{b^{\scriptscriptstyle{(2)}}}$.
\end{Lemma}
\begin{proof}
This fact is equivalent to the relations
\begin{equation*}
\begin{aligned}
&L_k^{\scriptscriptstyle{(1)}}f_{-1/2}|\vac\ra=0,\quad  L_k^{\scriptscriptstyle{(2)}}f_{-1/2}|\vac\ra=0, \quad \text{ for $k>0$}, \\
&L_0^{\scriptscriptstyle{(1)}}f_{-1/2}|\vac\ra=\Delta_{1,2}(b^{\scriptscriptstyle{(1)}})f_{-1/2}|\vac\ra, \quad L_0^{\scriptscriptstyle{(2)}}f_{-1/2}|\vac\ra=\Delta_{2,1}(b^{\scriptscriptstyle{(2)}})f_{-1/2}|\vac\ra,
\\ 
&\left((L_{-1}^{\scriptscriptstyle{(1)}})^2+(b^{\scriptscriptstyle{(1)}})^2L_{-2}^{\scriptscriptstyle{(1)}}\right)f_{-1/2}|\vac\ra=0, \quad \left((L_{-1}^{\scriptscriptstyle{(2)}})^2+(b^{\scriptscriptstyle{(2)}})^{-2}L_{-2}^{\scriptscriptstyle{(2)}}\right)f_{-1/2}|\vac\ra=0.
\end{aligned}
\end{equation*}
The relations can be checked directly by use of \eqref{Vir12}.
\end{proof}
Therefore the current $f(q)$ can be considered as a product $\phi_{12}^{\scriptscriptstyle{(1)}}\phi_{21}^{\scriptscriptstyle{(2)}}$ (notations from \cite{BPZ}). This fact will be discussed below in Remark \ref{rem:phi12}.

The character of the module $\mathbb{L}^b_{m,n}$ equals to $\mathrm{ch}(\mathbb{L}^b_{m,n})=(1-q^{mn})\mathrm{ch}(\pi_\Vir^{\Delta_{m,n}})$ since it is a quotient of the Verma module by the submodule generated by the singular vector on the level $mn$.

\begin{thm} \label{thm:Vac}
The module $\mathrm{Vac}$ is isomorphic to the sum of $\Vir \oplus \Vir$ modules
\[
\mathrm{Vac}\cong \bigoplus\limits_{m\in\mathbb{N}}\mathbb{L}_{1,m}^{b^{\scriptscriptstyle{(1)}}}\otimes \mathbb{L}_{m,1}^{b^{\scriptscriptstyle{(2)}}}.
\]
\end{thm}
\begin{proof} The vector $|\vac\ra$ has the highest weight $0=\Delta_{1,1}(b^{\scriptscriptstyle{(1)}})=\Delta_{1,1}(b^{\scriptscriptstyle{(2)}})$ for the both Virasoro subalgebras. Therefore this vector generates $\Vir\oplus\Vir$ submodule $\mathbb{L}_{1,1}^{b^{\scriptscriptstyle{(1)}}}\otimes \mathbb{L}_{1,1}^{b^{\scriptscriptstyle{(2)}}}$. It was proved above that the vector $f_{-1/2}|\vac\ra$ generates the module $\mathbb{L}_{1,2}^{b^{\scriptscriptstyle{(1)}}}\otimes \mathbb{L}_{2,1}^{b^{\scriptscriptstyle{(2)}}}$. 

The vacuum module $\mathrm{Vac}$ is a quotient of the Verma module $\pi^0_{\sfF\oplus \NSR}$. This Verma module has only one free field realization, namely we put $P=Q/2$ and use upper sign in formulae~\eqref{eq.FFR}. Therefore the vectors $|P,n\ra$ for $n \geq 0$ are well defined in the Verma module $\pi^0_{\sfF\oplus \NSR}$. It follows from formula~\eqref{eq:|P,n} that the vector $|P,n\rangle$ contains the product $\prod\limits_{r=1/2}^{(4n-1)/2}f_{-r}\overline{|P\ra}$ with non-zero coefficient. Then the vectors $|P,n\ra$ for $n \geq 0$, $2n \in \mathbb{Z}$ are non-zero in the  quotient module $\mathrm{Vac}$. Since $P=Q/2$ then  $P^{(\eta)}=Q^{(\eta)}/2$, $\eta=1,2$. Using \eqref{momentum} we get $P_n^{\scriptscriptstyle{(1)}}=P_{1,2n+1}(b^{\scriptscriptstyle{(1)}})$, $P_n^{\scriptscriptstyle{(2)}}=P_{2n+1,1}(b^{\scriptscriptstyle{(2)}})$.

So we proved that the vector $|P,n\ra \in \mathrm{Vac}$ generates the $\Vir\oplus\Vir$ submodule of the highest weight $(\Delta_{1,m}(b^{\scriptscriptstyle{(1)}}),\Delta_{m,1}(b^{\scriptscriptstyle{(2)}}))$, where $m=2n+1$. The irreducible module $\mathbb{L}_{1,m}^{b^{\scriptscriptstyle{(1)}}}\otimes \mathbb{L}_{m,1}^{b^{\scriptscriptstyle{(2)}}}$ is the smallest module of this highest weight. Therefore we get an inequality of characters
\begin{align}
\mathrm{ch}(\mathrm{Vac}) \geq \sum\limits_{m\in\mathbb{N}}\mathrm{ch}\left(\mathbb{L}_{1,m}^{b^{\scriptscriptstyle{(1)}}}\otimes \mathbb{L}_{m,1}^{b^{\scriptscriptstyle{(2)}}}\right). \label{eq:chVac}
\end{align}
Now our theorem is equivalent to equality in \eqref{eq:chVac}. So it remains to prove an identity
\[
(1-q^{1/2})\prod_{k=1}^\infty \frac{(1+q^{k-\frac12})^2}{1-q^k}= \sum_{m \in \mathbb{N}} q^{\frac{(m-1)^2}2}(1-q^m)^2\prod_{k=1}^\infty \frac{1}{(1-q^k)^2}
\]
Equivalently
\[
(1-q^{1/2})\prod_{k=1}^\infty (1+q^{k-\frac12})^2(1-q^k)= \sum_{m \in \mathbb{N}} q^{\frac{(m-1)^2}2}(1-q^m)^2.
\]
And the last identity is a $(1-q^{1/2})$ multiple of the Jacobi triple product identity \eqref{JTP} for $y=1$, $q\mapsto q^{1/2}$
\[
(1-q^{1/2})\prod_{k=1}^\infty (1+q^{k-\frac12})^2(1-q^k)= \sum_{m \in \mathbb{Z}} q^{\frac{m^2}2}(1-q^{1/2})=\sum_{m \in \mathbb{N}} q^{\frac{(m-1)^2}2}(1-q^m)^2.
\]
\end{proof}

\section{Vertex operators and conformal blocks}
\label{sec:conf:block}

\subsection{Conformal blocks and chain vectors}
\label{ssec:chain}

We use non hermitian, but a complex symmetric scalar product. Operators are conjugated as
\begin{equation}\label{eq:conj:operators}
L_{n}^+=L_{-n},\quad
G_{r}^+=G_{-r},\quad
f_{r}^+=-f_{-r} \quad\Rightarrow \quad (L^{(\eta)}_n)^+=L^{(\eta)}_{-n}, \quad \eta=1,2.
\end{equation}
We normalize the highest weight vectors of Verma modules by $\la \Delta|\Delta\ra=1$, $\overline{\langle\Delta^{\NS}}|\overline{\Delta^{\NS}\rangle}=1$. The coefficients $\Omega_n(P)$ are determined by the similar condition $\la P,n|P,n\ra=1$.

The vertex operator $V^{\D}_{\D_1,\D_2}\colon \pi^{\D_2}_{\Vir} \mapsto \pi^{\D_1}_{\Vir}$ is defined by the commutation relations
\begin{equation}
[L_k,V_{\D}(q)]=\left(q^{k+1}\partial_q +(k+1)\D q^{k}\right)V_{\D}(q), \label{vertcommrel}
\end{equation}
Here and below we simplify the notation $V^{\D}_{\D_1,\D_2}$ to $V_{\D}(q)$. It follows from \eqref{vertcommrel} that $V_\Delta(q)$ can be written as a power series $V_\Delta(q)=\sum_{m \in \mathbb{Z}} V_{\Delta,m}q^{-m+\Delta_1-\Delta_2-\Delta}$.
The operator $V_\Delta(q)$ is completely determined by the relation \eqref{vertcommrel} and normalization $\la \D_2|V^{\D}_{\D_2,\D_1}(1)| \D_1\ra=1$. We express the conformal weight of the vertex operator in terms of the parameter $\alpha$ (cf. \eqref{eq:cDelta:Vir})
\[
\Delta=\Delta(\alpha-Q/2,b)=\alpha(Q-\alpha), 
\]
and abbreviate $V_\Delta$ to $V_\alpha$. The $n$-point conformal block on $\mathbb{CP}^1$ of the primary fields located in the points $z_n=\infty, z_{n-1}, \ldots, z_{2}\in\mathbb{C}\backslash\{0\}, z_1=0; z_i\neq z_j, i\neq j $ is defined as the matrix element
\begin{equation}
\kr{F}_c(\{\D_{i-j}\},\{\D_k\}|\{z_k\})=\la\D_n|V^{\D_{n-1}}_{\D_n, \D_{(n-1)-(n-2)}}(z_{n-1}) \ldots V^{\D_{2}}_{\D_{3-2},\D_{1}}(z_{2})|\D_1\ra \label{confblock}
\end{equation}
Here $\D_i$ are the highest weights of Verma modules corresponding to the points $z_i$, $\D_{(i+1)-i}$ are the highest weights of intermediate Verma modules. The expression \eqref{confblock} defines the conformal block as a multivariable formal power series in $\frac{z_{i}}{z_{i+1}}$. This conformal block can be represented by use of the diagram in Figure \ref{fig:npcb}. It has been argued in \cite{Teschner:2004} that this power series converges in a region $\frac{z_{i}}{z_{i+1}}<<1$. In this paper (except Subsection \ref{ssec:melem}) we do not use analytical properties of conformal blocks (see also Remark \ref{rem:convergence}). We will consider variables $\frac{z_{i}}{z_{i+1}}$ as a formal variables. Due to convergence mentioned above it is equivalent to the study of vicinity of 0.

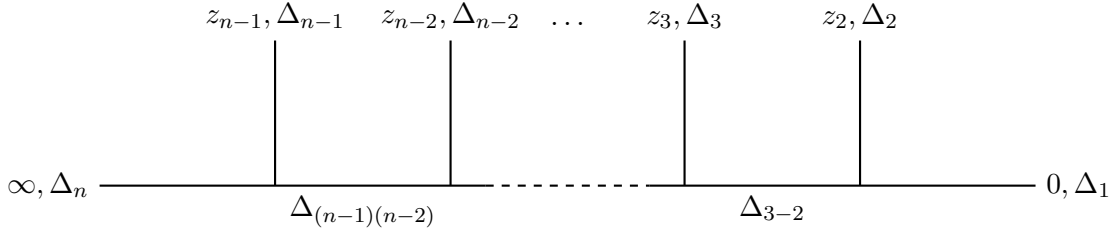
\begin{figure}
\begin{center}
\begin{tikzpicture}[x=4em,y=2em]
\draw[thick] (0.7,-1) -- (4,-1) node[anchor=west]{$0, \D_1$};
\draw[thick] (2.5,-1) -- (2.5,1.5) node[anchor=south]{$z_{2}, \D_{2}$};
\draw[thick] (1,-1) -- (1,1.5) node[anchor=south]{$z_{3}, \D_{3}$};
\draw[thick] (-1,-1) -- (-1,1.5) node[anchor=south]{$z_{n-2}, \D_{n-2}$};
\draw[thick] (-2.5,-1) -- (-2.5,1.5) node[anchor=south]{$z_{n-1}, \D_{n-1}$};
\draw[thick] (-0.7,-1) -- (-4,-1) node[anchor=east]{$\infty, \D_{n}$};
\draw[thick, dashed] (-0.7,-1) -- (0.7,-1);
\draw(1.75,-1) node[anchor=north]{$\D_{3-2}$};
\draw(-1.75,-1) node[anchor=north]{$\D_{(n{-}1)(n{-}2)}$};
\draw(0,2) node[anchor=north]{$\ldots$};
\end{tikzpicture}
\end{center}
\caption{\leftskip=1in \rightskip=1in
Diagram representing conformal block as a matrix element. \label{fig:npcb}}
\end{figure}

For the 4-point conformal block one can set the point $z_3$ to $1$, using the conformal transformation $z\mapsto z/z_3$. We define the 4-point conformal block by the formula
\begin{equation}\label{eq:Vir4point}
\kr{F}_c(\ora{\D},\D|q)=q^{\D_1+\D_2}\la\D_4|V^{\D_3}_{\D_4, \D}(1) V^{\D_2}_{\D, \D_1}(q)|\D_1\ra,
\end{equation}
where $\ora{\D}$ stands for the set of external weights $\D_i, i=\overline{1,4}$. Note that the function $\kr{F}_c$ differs from the function defined in equation \eqref{confblock} by a factor $q^{\D_1+\D_2}$.

It is convenient for calculations to rewrite the definition of $\kr{F}_c(\ora{\D},\D|q)$ in terms of the chain vectors $|W(q)\ra_{21}$ defined as
\[
|W(q^2)\ra_{21}=q^{\D_1+\D_2}V^{\D_2}_{\D, \D_1}(q)|\D_1\ra=q^{\D} \sum_{\lmb\in\mathbb{Y}} c^{21}_\lmb L_{-\lmb} q^{|\lmb|}|\D\ra, 
\]
where $\lambda=(\lambda_1\geq\lambda_2\geq \ldots\geq\lambda_k>0)$ is a partition, 
$L_{-\lmb}|\D\ra\equiv L_{-\lmb_1} \ldots L_{-\lmb_{k}}|\D\ra$. If we decompose $|W(q)\ra_{21}=q^{\D/2}\sum\limits_{N=0}^{\infty}q^{\frac{N}2}|N\ra_{21}$, then the commutation relation  \eqref{vertcommrel} implies 
\begin{equation}\label{eq:Vir_chain}
L_k|N\ra_{21}=\Bigl(k\D_2-\D_1+\D+N-k\Bigr)\,|N-k\ra_{21}, \quad  k>0.
\end{equation}
Here and below we assume that $|N\rangle=0$, for $N<0$. These equations coupled to the normalization $|0\ra_{21}=|\D\ra$ determine the chain vector.

It is easy to see that the conjugate vertex operator satisfies \eqref{vertcommrel} $V_{\D}^{+}(1/q)q^{-2\D}=V_{\D}(q)$. Therefore we define the conjugate chain vector as ${}_{34}\la W(1)|=\la\D_4|V^{\D_3}_{\D_4, \D}(1)$ and the conformal block  $\kr{F}_c$ can be written as
\[
\kr{F}_c(\overrightarrow{\D},\D|q)=_{34}\la W(1)|W(q^2)\ra_{21}=q^{\D}\sum\limits_{N=0}^{\infty}q^{N}{}_{34}\la N|N\ra_{21}={}_{34}\la W(q)|W(q)\ra_{21}, 
\]

We also use so-called irregular limit (other names Whittaker limit and Gaiotto limit) of the conformal blocks and chain vectors~\cite{Gaiotto}. Namely one can rescale the chain vector $|W(q)\ra_{21}$
\begin{equation}
|N\ra_{21}=(-\D_1)^N|N\ra_{21}',
\label{Virirrlim}
\end{equation}
and tend $\D_{1}$ to $\infty$. Then the equations \eqref{eq:Vir_chain} simplify to (we omit $'$ symbol in $|N\ra_{21}'$ below)
\[
L_1|N\ra=|N-1\ra, \quad N>0, \quad 
L_k|N\ra=0, k>1 ,
\]
or equivalently, in terms of the Whittaker vector (the limit of chain vector) $|W(q)\ra$
\[
L_1|W(q)\ra=q^{1/2}|W(q)\ra,\quad L_k|W(q)\ra=0, k>1.
\]
Note that it is enough to impose $L_1$ and $L_2$ relations since the action of the other $L_k$, $k>2$ follows from the Virasoro commutation relations. The irregular (or Whittaker, or Gaiotto) limit of conformal block is defined by
\begin{equation}\label{eq:Whitcom}
\kr{F}_c(\D|q)=\la W(q)|W(q)\ra=q^{\D}\sum\limits_{N=0}^{\infty}q^{N}\la N|N\ra.
\end{equation}

Now we consider $\mathcal{N}=1$ superconformal field theory (SCFT) with the $\NSR$ symmetry \cite{Belavin:2007zz,Belavin:2007eq,Hadasz:2006sb}. In this case we have a multiplet of two vertex operators: even $\Phi_{\D^{\NS}}(q)$ and  odd $\Psi_{\D^{\NS}}(q)=[G_{-1/2},\Phi_{\D^{\NS}}(q)]$. These operators act from a Verma module to a Verma module  \[\Phi^{\D^{\NS}}_{\D^{\NS}_2,\D^{\NS}_1}, \Psi^{\D^{\NS}}_{\D^{\NS}_2,\D^{\NS}_1}\colon \pi_{\NSR}^{\D_1^{\NS}}\mapsto \pi_{\NSR}^{\D_2^{\NS}}\]
and are determined by the commutation relations
\begin{equation}
\begin{aligned}
&[L_{k},\Phi_{\D^{\NS}}(q)] =
(q^{k+1}\partial_{q}+(k+1)\D^{\NS} q^{k})\Phi_{\D^{\NS}}(q),\\
& [L_{k},\Psi_{\D^{\NS}}(q)] =
(q^{k+1}\partial_{q}+(k+1)(\D^{\NS}+1/2)q^{k})\Psi_{\D^{\NS}}(q),\\
&[G_{r},\Phi_{\D^{\NS}}(q)] =
q^{r+1/2}\Psi_{\D^{\NS}}(q),\\
&\{G_{r},\Psi_{\D^{\NS}}(q)\}=
(q^{r+1/2}\partial_{q}+(2r+1)\D^{\NS}q^{r-1/2})\Phi_{\D^{\NS}}(q),\\
\end{aligned}\label{prim}
\end{equation}
We use the normalization $\la P_2|\Phi^{\D^{\NS}}_{\D^{\NS}_2,\D^{\NS}_1}(1)|P_1 \ra=\la P_2|\Psi^{\D^{\NS}}_{\D^{\NS}_2,\D^{\NS}_1}(1)|P_1 \ra=1.$ We express the conformal weight of the vertex operator in terms of the parameter $\alpha$
\begin{equation}
\Delta^{\NS}=\Delta^{\NS}(\alpha-Q/2,b)=\frac12\alpha(Q-\alpha)
\end{equation}
and abbreviate notation for vertex operators to $\Phi_\alpha$, $\Psi_\alpha$.

As in the previous case one can define the $n$-point $\mathbb{CP}^1$ conformal blocks by the formulae
\begin{equation}
\kr{F}_{c^{\NS}}(\{\D^{\NS}_{i-j}\}|\{\D^{\NS}_k\},\{z_k\})=\la\D^{\NS}_n||\Phi^{\D^{\NS}_{n-1}}_{\D^{\NS}_n, \D^{\NS}_{(n-1)-(n-2)}}(z_{n-1}) \ldots \Phi^{\D^{\NS}_{2}}_{\D^{\NS}_{3-2},\D^{\NS}_{1}}(z_{2})|\D^{\NS}_1\ra \label{eq:NSRconfblock}
\end{equation}
for the $\langle \Phi \Phi \ldots \Phi \rangle$ conformal blocks and similarly for conformal blocks containing $\Psi$ fields. The 4-point conformal block can be defined also by use of the chain vectors 
\begin{equation}
\begin{aligned}
|W_{\NS}(q^2)\ra_{21}&=q^{\D_1^{\NS}+\D_2^{\NS}}\left(\Phi^{\D_2^{\NS}}_{\D^{\NS},\D_1^{\NS}}(q)|\D_1^{\NS}\ra\right)=q^{\D^{\NS}}\left(\sum\limits_{2N=0}^{\infty}q^{N}|N\ra^{\NS}_{21}\right)\\
|\td{W_{\NS}(q^2)}\ra_{21}&=q^{\D_1^{\NS}+\D_2^{\NS}+\frac12}\left(\Psi^{\D_2^{\NS}}_{\D^{\NS},\D_1^{\NS}}(q)|\D_1^{\NS}\ra\right)= q^{\D^{\NS}}\left(\sum\limits_{2N=0}^{\infty}q^{N}\td{|N\ra}^{\NS}_{21}\right), \label{chain}
\end{aligned}
\end{equation}
where the index $N$ runs over integer and half-integer values. These chain vectors are determined by the recursion relations (which follow from \eqref{prim})
\begin{equation}
\begin{aligned}
&L_k|N\ra_{21}^{\NS}=(k\D_2^{\NS}-\D_1^{\NS}+\D^{\NS}+N-k)|N-k\ra_{21}^{\NS}\\
&L_k\td{|N\ra}_{21}^{\NS}=(k\D_2^{\NS}-\D_1^{\NS}+\D^{\NS}+N-k/2)|\td{N-k}\ra^{\NS}_{21}\\
&G_r\td{|N\ra}^{\NS}_{21}=(2r\D_2^{\NS}-\D_1^{\NS}+\D^{\NS}+N-r)|N-r\ra^{\NS}_{21}\\
&G_r|N\ra^{\NS}_{21}=|\td{N-r}\ra^{\NS}_{21},\quad 2N\in\mathbb{Z},\quad N>k>0 \label{ChAx}
\end{aligned}
\end{equation}
combined with the normalization $|0\ra^{\NS}_{21}=|\tilde{0}\ra^{\NS}_{21}=|\Delta^{\NS}\rangle$. In terms of the chain vectors we have
\begin{equation}
\begin{aligned}
L_k |W_{\NS}(q)\ra_{21}&=\Bigl(k\D_2^{\NS}-\D_1^{\NS}+L_0\Bigr)q^{\frac{k}2}|W_{\NS}(q)\ra_{21},& L_k |\td{W_{\NS}}(q)\ra_{21}&=\left(k\D_2^{\NS}-\D_1^{\NS}+L_0+\frac{k}{2}\right)q^{\frac{k}2}|\td{W_{\NS}}(q)\ra_{21},
 \\
G_r|W_{\NS}(q)\ra_{21}&=q^{\frac{r}2}|\td{W_{\NS}}(q)\ra_{21}   & G_r|\td{W_{\NS}}(q)\ra_{21}&=\Bigl(2r\D_2^{\NS}-\D_1^{\NS}+L_0\Bigr)q^{\frac{r}2}|W_{\NS}(q)\ra_{21} \label{eq:W:NS} 
\end{aligned}
\end{equation}

It follows from \eqref{prim} that operators 
$\Phi^{+}_{\D^{\NS}}(1/q)q^{-2\D^{\NS}}$ and $-q^{-2\D^{\NS}-1}\Psi^{+}(1/q)$ have the same commutation relation as $\Phi_{\D^{\NS}}(q)$ and $\Psi_{\D^{\NS}}(q)$. But their normalization differs due to the minus sign at $\Psi^+$. Therefore we have the equality of matrix elements
\begin{equation}
\la w_2|\Phi_{\alpha}(q)|w_1\ra=q^{-2\Delta^{\NS}}(-1)^{2(n_1+n_2)}\la w_1|\Phi_{\alpha}(1/q)|w_2\ra
\label{eq:conj.vertex}
\end{equation}
where $w_1 \in \pi_{\NSR}^{\Delta_1^{\NS}}$, $w_2 \in \pi_{\NSR}^{\Delta_1^{\NS}}$, $L_0w_1=(\Delta_1^{\NS}+n_1)w_1$, $L_0w_2=(\Delta_2^{\NS}+n_2)w_2$.  

Therefore conjugate chain differs from the by signs and we can define two $\NSR$ 4-point conformal blocks (namely the $\la\Phi\Phi\Phi\Phi\ra$ and the $\la\Phi\Psi\Psi\Phi\ra$ conformal blocks)
\begin{equation}\label{eq:FNS:4p}
\kr{F}_{c_{\NS}}(\overrightarrow{\D^{\NS}},\D^{\NS}|q)=_{34}\la W_{\NS}(q)|W_{\NS}(q)\ra_{21}, \quad \td{\kr{F}}_{c_{\NS}}(\overrightarrow{\D^{\NS}},\D^{\NS}|q)=_{34}\td{\la W_{\NS}}(q)|\td{W_{\NS}}(q)\ra_{21}
\end{equation}
Similarly to the Virasoro case the functions $\kr{F}_{c_{\NS}}$ and $\td{\kr{F}}_{c_{\NS}}$ differ from the conformal blocks defined in \eqref{eq:NSRconfblock} by factors $q^{\D_1^{\NS}+\D_2^{\NS}}$ and $q^{\D_1^{\NS}+\D_2^{\NS}+\frac12}$ correspondingly. 

\begin{Remark} The vectors $|N\rangle^{\NS}$ with integer and half-integer $N$ do not interact in relation \eqref{ChAx}. In particular if we change the normalization of $|0\rangle^{\NS}_{21}$  then we change the coefficients at the integer powers of $q$ (more precisely $q^{\Delta^{NS}+N}$, $N \in \mathbb{Z}$) in $\kr{F}_{c_{\NS}}$ but the coefficients at the half-integer powers of $q$ remain unchanged. Conversely if we rescale $\td{|N\ra}^{\NS}_{21}$ then we rescale the coefficients at the half integer powers of $q$ in $\kr{F}_{c_{\NS}}$ but the coefficients at the integer powers of $q$ remain unchanged.
\end{Remark}

Irregular (or Whittaker) limit of the $\NSR$ conformal blocks and chain vectors is defined as follows. One can rescale the $\NSR$ chains
\begin{equation}
\begin{aligned}
|N\ra_{21}^{\NS}&=(-\D_1^{\NS})^N|N\ra_{21}^{'\NS},\;\; N \in \mathbb{Z},&\qquad |N\ra_{21}^{\NS}&=(-\D_1^{\NS})^{N-1/2}|N\ra_{21}^{'\NS},\;\; N \in \mathbb{Z}+\frac12,\quad\\
\td{|{N}\ra}_{21}^{\NS}&=(-\D_1^{\NS})^{N}\td{|N\ra}^{'\NS}_{21},\;\; N \in \mathbb{Z},&\qquad \td{|N\ra}_{21}^{\NS}&=(-\D_1^{\NS})^{N+1/2}\td{|N\ra}^{'\NS}_{21},\;\; N \in \mathbb{Z}+\frac12, \label{irrlim}
\end{aligned} 
\end{equation}
and tend $\D^{\NS}_{1}$ to $\infty$. In the limit both tilded and non-tilded chain vectors coincide and obey the relation
\[
G_{1/2}|N\ra^{\NS}=|N-1/2\ra^{\NS}, \quad 
G_{3/2}|N\ra^{\NS}=0. 
\]
The formulae for action of $L_k$, $k>0$ and $G_r$, $r>3/2$ follows from the $\NSR$ commutation relations \eqref{eq:FNSRrelations}. In terms of the Whittaker vector $|W_{\NS}(q)\ra=q^{\D^{\NS}/2}\sum_{N=0}^{\infty}q^{\frac{N}2}|N\ra^{\NS}$ the action of $G_{1/2}, G_{3/2}$ can be written as 
\[G_{1/2}|W_{\NS}(q)\ra=q^{1/4}|W_{\NS}(q)\ra,\quad  G_{3/2}|W_{\NS}(q)\ra=0.\] The irregular limit of conformal block is defined by the formula
\[
\kr{F}_{c_{\NS}}(\D^{\NS}|q)=\la W_{\NS}(q)|W_{\NS}(q)\rangle
\]

\subsection{$\Vir \oplus \Vir$ decomposition of chain vectors and vertex operators}
\label{ssec:chdec}

The $\sfF \oplus \NSR$ Whittaker vector is defined as a tensor product of the $\sfF$ vacuum $1 \in \pi_\sfF$ and the $\NSR$ Whittaker vector $|W_{\NS}(q)\ra\in \pi_\NSR$.  The decomposition of the $\sfF\oplus \NSR$ module  \eqref{FxNSR} provides a decomposition of the corresponding Whittaker vector
\[
|1\otimes W_{\NS}(q)\rangle=\sum_{2n\in\mathbb{Z}}|v_n(q)\ra,
\]
where $|v_n(q)\ra\in\pi^n_{\Vir\oplus \Vir}$. It turns out that $|v(q)\ra_n$ is the Whittaker vector for the algebra $\Vir\oplus\Vir$
\begin{prop} The Whittaker vector for the $\sfF\oplus \NSR$ algebra equals to the sum of $\Vir\oplus\Vir$ Whittaker vectors
\begin{align}
|1\otimes W_{\NS}(q)\ra=\sum\limits_{2n\in\mathbb{Z}}
\left(l_{n}(P,b)\,\Bigl(|W_{n}^{\scriptscriptstyle{(1)}}(\beta^{\scriptscriptstyle{(1)}}q)\ra\otimes|W_{n}^{\scriptscriptstyle{(2)}}(\beta^{\scriptscriptstyle{(2)}}q)\ra\Bigr)\right) \label{Whitdecomp}.
\end{align}
Here $|W_n^{\scriptscriptstyle{(1)}}\ra\otimes|W_n^{\scriptscriptstyle{(2)}}\ra$ denotes the tensor product of Whittaker vectors in $\pi^n_{\Vir\oplus\Vir}$, and the coefficients $l_n(P,b)$ do not depend on $q$. The parameters $\beta^{(\eta)}, \eta=1,2$ are defined by the formulae
\begin{align}\label{eq:beta12}
\beta^{\scriptscriptstyle{(1)}}=\left(\frac{b^{-1}}{b^{-1}-b}\right)^2,
\quad \beta^{\scriptscriptstyle{(2)}}=\left(\frac{b}{b-b^{-1}}\right)^2.
\end{align}
\label{prop:Whitdecomp}
\end{prop}
The values of $l_n(P,b)$ will be computed in the next subsection.

\begin{proof}
Let us act $L^{(\eta)}_{1}, L^{(\eta)}_{2},  \eta=1,2$ on $|1\otimes W_{\NS}(q)\ra$. Using expressions \eqref{Vir12} and linear independence of vectors from different $\pi^n_{\Vir\oplus\Vir}$ we have
\begin{equation*}
\begin{aligned}
L^{\scriptscriptstyle{(1)}}_{1}|v_n(q)\ra&=(\beta^{\scriptscriptstyle{(1)}} q)^{1/2}|v_n(q)\ra,& \quad L^{\scriptscriptstyle{(2)}}_{1}|v_n(q)\ra&=(\beta^{\scriptscriptstyle{(2)}} q)^{1/2}|v_n(q)\ra, \\
L^{\scriptscriptstyle{(1)}}_{2}|v_n(q)\ra&=0,& \quad
L^{\scriptscriptstyle{(2)}}_{2}|v_n(q)\ra&=0.
\end{aligned}
\end{equation*}
Therefore the vector $|v_n(q)\ra$ is proportional to the tensor products of Whittaker vectors. Hence we proved~\eqref{Whitdecomp}.
\end{proof}

There is an analogous decomposition of the chain vector.
\begin{prop} The $\sfF\oplus \NSR$ chain vector equals to the sum of $\Vir\oplus\Vir$ chain vectors 
\begin{align}
|1\otimes W_{\NS}\ra_{21}=\sum\limits_{2n\in\mathbb{Z}}\left(l_n^{21}(P,b,|\D^{\NS}_2,\D^{\NS}_1) \Bigl(|W^{\scriptscriptstyle{(1)}}_{n}\ra_{21}(q)\otimes|W^{\scriptscriptstyle{(2)}}_{n}\ra_{21}(q)\Bigr)\right). \label{Chdecomp}
\end{align}
The external weights of the $\Vir$ chain vectors $|W^{\scriptscriptstyle{(1)}}_{n}\rangle_{21}$, $|W^{\scriptscriptstyle{(2)}}_{n}\rangle_{21}$ are related to $\Delta_1^{\NS}, \Delta_2^{\NS}$ by \eqref{eq:DeltaVirNSR}
\[
\D_{i}^{\scriptscriptstyle{(1)}}=\frac{b^{-1}}{b^{-1}-b}\D_{i}^{\NS}, \quad \D_{i}^{\scriptscriptstyle{(2)}}=\frac{b}{b-b^{-1}}\D_{i}^{\NS}. 
\]
\label{prop:Chdecomp}
\end{prop}
\begin{proof}
The proof is similar to the previous one but the computations are more cumbersome. We act $L_k^{\scriptscriptstyle{(1)}}$ on $|1 \otimes W_{\NS}\ra_{21}$ using \eqref{Vir12} and \eqref{eq:W:NS}
\begin{multline*}
L_k^{\scriptscriptstyle{(1)}}|1\otimes W_{\NS}\ra_{21}=\frac{b^{-1}}{b^{-1}-b}q^{k/2}|1\ra \otimes(k\D_2^{\NS}-\D_1^{\NS}+L_0)|W_{\NS}\ra_{21}+\frac{1}{b^{-1}-b} \sum_{r\in\mathbb{Z}_{\geq 0}+\frac12}f_{-r}G_{r+k}|1\otimes W_{\NS}\ra_{21}=
\\ 
=q^{k/2}\left(\frac{b^{-1}}{b^{-1}-b}(k\D_2^{\NS}-\D_1^{\NS}+L_0) |1\otimes W_{\NS}\ra_{21}+\frac{1}{b^{-1}-b} \sum_{r\in\mathbb{Z}_{\geq 0}+\frac12}f_{-r}G_{r}|1\otimes W_{\NS}\ra_{21}\right)=
\\
 = q^{k/2}\left(k\D_2^{\scriptscriptstyle{(1)}}-\D_1^{\scriptscriptstyle{(1)}}+L_0^{(1)}\right)|1\otimes W_{\NS}\ra_{21}, 
\end{multline*}
On the other hand the relations \eqref{eq:Vir_chain} are equivalent to 
\[
L_k^{\scriptscriptstyle{(1)}}|W^{\scriptscriptstyle{(1)}}_{n}\ra_{21}=q^{k/2}\left(k\D_2^{\scriptscriptstyle{(1)}}-\D_1^{\scriptscriptstyle{(1)}}+L_0^{\scriptscriptstyle{(1)}}\right)|W_{n}^{\scriptscriptstyle{(1)}}\ra_{21}.
\]
The calculation for $L_k^{\scriptscriptstyle{(2)}}$ is similar.
\end{proof}
\begin{Remark}
Note that the additional factors $\beta^{(\eta)}, \eta=1,2$ in \eqref{Whitdecomp} do not appear in \eqref{Chdecomp}. These factors are an artefact of the irregular limit.
\end{Remark}


Consider now the vertex operator $\overline{\Phi}_\alpha= 1 \otimes \Phi_\alpha$ acting from one $\sfF\oplus \NSR$ Verma module $\pi^{\D'^{\NS}}_{\sfF\oplus \NSR}$ to another one $\pi^{\D^{\NS}}_{\sfF\oplus \NSR}$. Due to decomposition \eqref{FxNSR} one can restrict $\overline{\Phi}_\alpha$ to a map between submodules $\pi^{n'}_{\Vir\oplus\Vir}$ and $\pi^{n}_{\Vir\oplus\Vir}$ for each $n', n$, such that $2n',2n \in\mathbb{Z}$

\begin{thm} The restriction $\sfF\oplus \NSR$ vertex operator in terms of the $\Vir\oplus\Vir$ subalgebra has the form
\begin{equation} 
\overline{\Phi}_{\alpha}(q)|_{\pi^{n'}_{\Vir\oplus\Vir}\rightarrow\pi^{n}_{\Vir\oplus\Vir}}=l_{n n'}(P,\alpha,P') \left( V_{\alpha^{\scriptscriptstyle{(1)}}}(q) \otimes V_{\alpha^{\scriptscriptstyle{(2)}}}(q) \right)\label{spd},
\end{equation}
where 
\[
   \alpha^{\scriptscriptstyle{(1)}}=\frac{\alpha}{\sqrt{2-2b^{2}}},\quad
   \alpha^{\scriptscriptstyle{(2)}}=\frac{\alpha}{\sqrt{2-2b^{-2}}}; 
\]
and $|_{\pi^{n'}_{\Vir\oplus\Vir}\rightarrow\pi^{n}_{\Vir\oplus\Vir}}$ means the restriction to the map between these submodules.
\label{thm:decomp}
\end{thm}

It follows from the definition of the chain vector  \eqref{chain} that Proposition \ref{prop:Whitdecomp} follows from this theorem. 

The Theorem \ref{thm:decomp} was stated in \cite{BBFLT} without proof. Our proof here is standard and similar to the one in \cite{HJ}.
\begin{proof}
Define $l_{n n'}(P,\alpha,P')$ as a quotient
\begin{equation}
l_{n n'}(P,\alpha,P')=\frac{\langle P,n|\overline{\Phi}_{\alpha}(q)|P',n'\ra}{\la P,n|V_{\alpha^{\scriptscriptstyle{(1)}}}(q)V_{\alpha^{\scriptscriptstyle{(2)}}}(q)|P',n'\rangle}. \label{eq:def:lnn}
\end{equation}
First we prove that $l_{n n'}(P,\alpha,P')$ does not depend on $q$. Recall that $T_f(q)=\frac{1}{2}:f'(q)f(q):=\sum L_n^{f}q^{-n-2}$ is the fermion stress-energy tensor. We act by operator $L_0^{\scriptscriptstyle{(1)}}+L_0^{\scriptscriptstyle{(2)}}=L_0+L_0^{f}$ and get the equation
\begin{equation*}
\begin{aligned}
\Bigl(\D^{\NS}(P
)+2n^2\Bigr)\la P,n|\overline{\Phi}_{\alpha}(q)|P',n'\ra=\la P,n|\left(L_0+L_0^{f}\right)\overline{\Phi}_{\alpha}(q)|P',n'\ra=\\=q\frac{d}{dq}\la P,n|\overline{\Phi}_{\alpha}(q)|P',n'\ra+\Bigl(\D^{\NS}(\alpha-Q/2)+\D^{\NS}(P')+2n'^2\Bigr)\la P,n|\overline{\Phi}_{\alpha}(q)|P',n'\ra.
\end{aligned}
\end{equation*}
Therefore $\la P,n|\overline{\Phi}_{\alpha}(q)|P',n'\ra\sim q^{\D^{\NS}(P)+2n^2-\D^{\NS}(\alpha-Q/2)-\D^{\NS}(P')-2n'^2}$. Similarly we have
\begin{multline*}
(\D^{\NS}(P)+2n^2)\la P,n|(V_{\alpha^{\scriptscriptstyle{(1)}}}V_{\alpha^{\scriptscriptstyle{(2)}}})(q)|P',n'\ra=\la P,n|\left(L_0^{\scriptscriptstyle{(1)}}+L_0^{\scriptscriptstyle{(2)}}\right)(V_{\alpha^{\scriptscriptstyle{(1)}}}V_{\alpha^{\scriptscriptstyle{(2)}}})(q)|P',n'\ra=
\\
=q\frac{d}{dq}\la P,n|(V_{\alpha^{\scriptscriptstyle{(1)}}}V_{\alpha^{\scriptscriptstyle{(2)}}})(q)|P',n'\ra+
\\
+ \Bigl(\D(\alpha^{\scriptscriptstyle{(1)}}-Q^{\scriptscriptstyle{(1)}}/2)+
 \D(\alpha^{\scriptscriptstyle{(2)}}-Q^{\scriptscriptstyle{(2)}}/2)+\D^{\NS}(P')+2n'^2\Bigr)\la P,n|(V_{\alpha^{\scriptscriptstyle{(1)}}}V_{\alpha^{\scriptscriptstyle{(2)}}})(q)|P',n'\ra.
\end{multline*}
Since $\D(\alpha^{\scriptscriptstyle{(1)}}-Q^{\scriptscriptstyle{(1)}}/2)+
 \D(\alpha^{\scriptscriptstyle{(2)}}-Q^{\scriptscriptstyle{(2)}}/2)=\D^{\NS}(\alpha-Q/2)$ we have $\la P,n|(V_{\D^{\scriptscriptstyle{(1)}}}V_{\D^{\scriptscriptstyle{(2)}}})(q)|P',n'\ra\sim q^{\D^{\NS}(P)+2n^2-\D^{\NS}(\alpha-Q/2)-\D^{\NS}(P')-2n'^2}.$ Therefore $l_{n n'}(P,\alpha,P')$ does not depend on $q$. 

Using the normalization of the vertex operators $V_{\D^{\scriptscriptstyle{(1)}}},V_{\D^{\scriptscriptstyle{(2)}}}$ we have $l_{n n'}(P,\alpha,P')=\langle P,n|\overline{\Phi}_{\alpha}(1)|P',n'\ra$.

Relation \eqref{spd} is equivalent to a relation for  matrix elements
\begin{align}
\left\langle P,n\left|L^{\scriptscriptstyle{(1)}}_{\lmb_2}L^{\scriptscriptstyle{(2)}}_{\mu_2}\left|\overline{\Phi}_{\alpha}(q)\right|L^{\scriptscriptstyle{(2)}}_{-\mu_1}L^{\scriptscriptstyle{(1)}}_{-\lmb_1}\right|P',n'\right\rangle=
l_{nn'}\cdot \left\langle P,n\left|L^{\scriptscriptstyle{(1)}}_{\lmb_2}L^{\scriptscriptstyle{(2)}}_{\mu_2}\left|V_{\alpha^{\scriptscriptstyle{(1)}}}(q)V_{\alpha^{\scriptscriptstyle{(2)}}}(q)\right|L^{\scriptscriptstyle{(2)}}_{-\mu_1}L^{\scriptscriptstyle{(1)}}_{-\lmb_1}\right|P',n'\right\rangle, \label{melemdec}
\end{align}
for any partitions $\lambda_1, \lambda_2, \mu_1, \mu_2$. Here we omit arguments $P, \alpha, P'$ in $l_{nn'}$.

We reshuffle the operators $L^{\scriptscriptstyle{(1)}}$ and $L^{\scriptscriptstyle{(2)}}$ to other sides of these matrix elements. Let us define the operators $\overline{\Phi}_{k l}, V_{k l}(q)$ as commutators with $L_0^{\scriptscriptstyle{(1)}}$ and $L_0^{\scriptscriptstyle{(2)}}$
\[
\overline{\Phi}_{kl}(q)=\mathrm{ad}_{L^{\scriptscriptstyle{(2)}}_{0}}^l\mathrm{ad}_{L^{\scriptscriptstyle{(1)}}_{0}}^k \left(\overline{\Phi}_{\alpha}(q)\right),\qquad  
V_{kl}(q)= \mathrm{ad}_{L^{\scriptscriptstyle{(2)}}_{0}}^l\mathrm{ad}_{L^{\scriptscriptstyle{(1)}}_{0}}^k \left( V_{\alpha^{\scriptscriptstyle{(1)}}}V_{\alpha^{\scriptscriptstyle{(2)}}}\right),
\]
where $k,l\geq 0$. Due to commutation relations \eqref{prim} the commutators $[L_m^{(\eta)},\overline{\Phi}_{k l}(q)], \eta=1,2$ turns out to be a sum of operators $\overline{\Phi}_{k l}$, $\overline{\Phi}_{k+1\, l}$, $\overline{\Phi}_{k\, l+1}$ and similarly for $V_{k l}(q)$. Moreover, the commutation relations between Virasoro generators $L^{(\eta)}_m$ and operators $\overline{\Phi}_{k l}(q)$ are equivalent to the commutation relations between $L^{(\eta)}_m$ and operators $V_{k l}(q)$. The last statement can be easily checked for $\overline{\Phi}_{0 0}$ and~$V_{0 0}$ 
\begin{equation*}
\begin{aligned}
{}[L^{\scriptscriptstyle{(1)}}_m,\overline{\Phi}_{0 0}(q)]=&\frac{b^{-1}}{b^{-1}-b}\left(q^{m+1}\partial_q+(m+1)\D^{\NS}(\alpha-Q/2)q^m\right)\overline{\Phi}_{\alpha}(q)+\frac{1}{b^{-1}-b}\sum_{r\in\mathbb{Z}+1/2} f_{m-r} q^{r+1/2} \overline{\Psi}_{\alpha}(q)=\\
=&q^{m}\overline{\Phi}_{1 0}(q)+m\D^{\scriptscriptstyle{(1)}}q^m\overline{\Phi}_{0 0}(q)\\
[L^{\scriptscriptstyle{(1)}}_m,V_{0 0}(q)]=&q^{m}V_{1 0}(q)+m\D^{\scriptscriptstyle{(1)}}q^mV_{0 0}(q),
\end{aligned}
\end{equation*} 
and similarly for $L^{\scriptscriptstyle{(2)}}_m$. The commutation relations between the other $\overline{\Phi}_{k l}(q), V_{k l}(q)$ and $L^{(\eta)}_m$ follow from previous relations and the commutation relations between $L^{(\eta)}_m$ and $L^{(\eta)}_0$.

Therefore the relation \eqref{melemdec} can be rewritten as
\[
\sum g_{kl} \la P,n|\overline{\Phi}_{k l}(q)|P',n'\ra=l_{nn'}\cdot\sum g_{kl}\la P,n|V_{k l}(q)|P',n'\ra,
\]
where the coefficients $g_{kl}$ on the left hand side and 
right hand side are equal. It remains to show that 
\[
\la P,n|\overline{\Phi}_{k l}(q)|P',n'\ra=l_{nn'}\cdot\la P,n|V_{k l}(q)|P',n'\ra,
\]
for any $k,l \geq 0$. Since $\la P,n|$ and $|P',n'\ra$ are eigenvectors for $L_0^{(\eta)}, \eta=1,2$, the last equation follows from the $k=0,l=0$ case, i.e. from the definition of $l_{nn'}$ \eqref{eq:def:lnn}.
\end{proof}

We will calculate $l_{n n'}$ in the next subsection. Now we note that the coefficients $l_n^{21}$ in  \eqref{Chdecomp} are equal to $l_{n0}$
\[
l_n^{21}(P,b|\D^{\NS}_2,\D^{\NS}_1)=\la P,n|\overline{\Phi}_{\alpha_2}(1)|P_1\ra=l_{n0}(P,\alpha_2,P_1),
\]
where $\alpha_2=P_2+Q/2$. The coefficients $l_n(P,b)$ of decomposition \eqref{Whitdecomp} are given by the irregular limit of~\eqref{Chdecomp}
\begin{align}
l_n(P,b)=(\beta^{\scriptscriptstyle{(1)}})^{-\D^{\scriptscriptstyle{(1)}}_n/2}(\beta^{\scriptscriptstyle{(2)}})^{-\D^{\scriptscriptstyle{(2)}}_n/2}\lim\limits_{\D_1^{\NS}\rightarrow \infty}\frac{l_n^{21}(P,b|\Delta_2,\Delta_1)}{(-\D_1^{\NS})^{\lfloor 2n^2 \rfloor}} \label{lirrlim},
\end{align}
where we used \eqref{irrlim} and \eqref{Virirrlim}.

\begin{Remark} \label{rem:phi12}
Recall that $\phi_{m,n}(q)$, $m,n \in \mathbb{N}$ denotes (\cite{BPZ}) the degenerate vertex operator for Virasoro CFT with
\[\alpha_{mn}=\frac12((m-1)b^{-1}+(n-1)b),\]
which satisfies an additional differential equation of order $mn$. For the operators $\phi_{1,1}(q), \phi_{2,1}(q), \phi_{1,2}(q)$ the corresponding equations read
\begin{equation}
\partial_q \phi_{1,1}(q)=0,\qquad \partial_q^2 \phi_{1,2}(q)+b^2 :\!T(q)\phi_{1,2}(q)\!:=0,\qquad \partial_q^2 \phi_{2,1}(q)+b^{-2} :\!T(q)\phi_{1,2}(q)\!:=0. \label{eq:phi_eq}
\end{equation}
The conformal weight of the operator $\phi_{m,n}(q)$ is denoted by $\Delta_{m,n}(b)=\Delta(\alpha_{m,n}-Q/2,b)$.

Similarly to Theorem \ref{thm:decomp} one can prove, that 
the matrix elements of $f(q)$ are proportional to the matrix elements of the product of $\phi_{1,2}^{\scriptscriptstyle{(1)}}(q)\phi_{2,1}^{\scriptscriptstyle{(2)}}(q)$. This fact is the operator analogue of the Lemma \ref{le:f-12}.

Moreover, due to the operator-state correspondence the highest weight vectors $|Q/2,n\ra \in \mathrm{Vac}$ correspond to the currents $\phi_{1,m}^{\scriptscriptstyle{(1)}}\phi_{m,1}^{\scriptscriptstyle{(2)}}$, $m=2n+1$.  
From the fusion rules \cite{BPZ} follow that the action of this current on the vector $\overline{|P\ra}$ shifts its momentum $(P_1,P_2) \rightarrow \left(P_1+\frac{k_1}2b^{\scriptsize(1)},P_2+\frac{k_2}2(b^{\scriptsize(2)})^{-1}\right)$, where $|k_1|,|k_2|<m$ and $m-k_1,m-k_2$ are odd.  Proposition \ref{pr:FNSR} states that only shifts with $k_1=k_2$ are allowed in our representations.

\end{Remark}

\subsection{Matrix elements $\la P,n |\Phi_{\alpha}(1) |P',n'\ra$}
\label{ssec:melem}
In this subsection we calculate the matrix elements $l_{n n'}(P,\alpha,P')=\la P,n|\Phi_{\alpha}(1)|P',n'\ra$. Introduce the functions $s_{\textrm{even}}(x,n)$ for $n \in \mathbb{Z}$ and $s_{\textrm{odd}}(x,n)$ for $n \in \mathbb{Z}+\frac{1}{2}$ by the formulae
\[
s_{\textrm{even}}(x,n)=
  \prod_{\substack{i,j\geq 0,\;i+j<2n\\ i+j\equiv0\bmod 2}}\hspace*{-10pt}(x+ib+jb^{-1}),\qquad
  s_{\textrm{odd}}(x,n)=2^{1/8}\hspace*{-10pt}
  \prod_{\substack{i,j\geq0,\;i+j<2n\\i+j\equiv1\bmod 2}}\hspace*{-10pt}(x+ib+jb^{-1}),
\]
for $n\geq 0$ and 
\[
   s_{\textrm{even}}(x,n)=(-1)^n s_{\textrm{even}}(Q-x,-n),\quad
   s_{\textrm{odd}}(x,n)=s_{\textrm{odd}}(Q-x,-n)
\]
for $n<0$. 

Recall that the vectors $|P,n\ra$ were fixed in \eqref{eq:|P,n}. The factors $\Omega_n$ were defined by the normalization condition. In the following theorem we give explicit expressions for them.

\begin{thm}\label{thm:ln} The matrix elements $ l_{n n'}(P,\alpha,P')$ have the form
\begin{multline}
l_{n n'}(P,\alpha,P')=\frac{ (-1)^{sg} 2^{n+n'}\Omega_n(P)\Omega_{n'}(P')}{ s_{\textrm{even}}(2P+Q|2n) s_{\textrm{even}}(2P'+Q|2n')}\times \\ \times
  \begin{cases}
  \prod_{\epsilon,\epsilon'=\pm}
  s_{\textrm{even}}(\alpha+\epsilon P'+\epsilon' P,\epsilon n'+\epsilon' n), \quad n+n' \in \mathbb{Z}\\
  \prod_{\epsilon,\epsilon'=\pm}
  s_{\textrm{odd}}(\alpha+\epsilon P'+\epsilon' P,\epsilon n'+\epsilon' n),
  \quad n+n'\in\mathbb{Z}+1/2
  \end{cases} 
\label{melemans}
\end{multline}
where the factors $\Omega_n(P)$ are given by the expressions
\begin{equation}
\Omega^2_n(P)=\frac{(-1)^{2n}\prod\limits_{i,j\geq 1,\;i+j=4n}(2P+ib+jb^{-1})}{2^{2n}\cdot 2P\prod\limits_{i=1}^{2n-1}(2P+2ib)\prod\limits_{j=1}^{2n-1}(2P+2jb^{-1})},\label{eq:Omega}
\end{equation}
and the sign factors are equal
\[
  (-1)^{sg}=\begin{cases}
  -1, n'\in\mathbb{Z}+1/2, n\in\mathbb{Z}\\
  1, \quad \text{otherwise}
  \end{cases}
\]
\end{thm}

These expressions were given in \cite{BBFLT} without proof (see also Remark \ref{rem:Omega}). The arguments in \cite{BBFLT} were partially based on the conjectural expression of $l_{n n'}^2$ in terms of the Liouville and super Liouville three point correlation functions. Here we find $l_{n n'}$ mimicking the standard approach to the Liouville three point function \cite{ZZbook},\cite{JTL} based on the associativity property and the properties of degenerate vertex operators $\phi_{12}$. Note that contrary to Liouville theory which is the coupling of chiral and antichiral CFT, here we use the product of two chiral CFT with the central charges $c^{\scriptscriptstyle{(1)}}$ and $c^{\scriptscriptstyle{(2)}}$.

As was already mentioned in the Introduction the formula \eqref{melemans} was also proven in \cite{HJ} by a different method based on the Dotsenko-Fateev integral representation of the conformal blocks. 

Remark that the explicit expressions \eqref{melemans} have the geometric meaning in the framework of the instanton counting.  We recall this in Section~\ref{sec:geom}. 

\begin{proof} 

Due to the symmetry $|P,n\ra=|-P,-n\ra$ it is sufficient to consider only a case when $n,n'\geq 0$.

In the proof we consider the matrix element $\la P,n|\overline{\Phi}_{\alpha}(1)f(q)|P',n'\ra$. This matrix element is a Laurent polynomial since $[\overline{\Phi}_{\alpha},f_r]=0$ and $f_r |P',n'\ra=0$ for $r\gg 0$, and $\la P,n|f_r=0$ for $r\ll 0$. Therefore we can consider this matrix element as an analytic function on $q\in \mathbb{CP}^1\setminus\{0,\infty\}$.

We decompose the proof into several steps.

\noindent \textbf{Step 1.} First we consider $\la P,n|\overline{\Phi}_{\alpha}(1)f(q)|P',n'\ra$ as a function on $|q|\ll 1$.

Using the decomposition \eqref{FxNSR} we can write $f(q)|P',n'\rangle=\sum_{s \in \mathbb{Z}}|u_s(q)\rangle$, where $|u_s(q)\rangle \in \pi_{\Vir\oplus \Vir}^{n'+s/2}$. Then we have a decomposition
\[
\la P,n|\overline{\Phi}_{\alpha}(1)f(q)|P',n'\ra=\sum_{s\in \mathbb{Z}}  \la P,n|\overline{\Phi}_{\alpha}(1)|u_s(q) \ra.
\]
Due to Theorem \ref{thm:decomp} and Remark \ref{rem:phi12} the matrix elements $\la P,n|\overline{\Phi}_{\alpha}(1)|u_s(q) \ra$ are proportional to conformal blocks of the CFT with $\Vir\oplus \Vir$ symmetry. The $\Vir\oplus \Vir$ conformal blocks factors to the product of the two $\Vir$ conformal blocks \eqref{eq:Vir4point} 
\[
\la P,n|\overline{\Phi}_{\alpha}(1)|u_s(q) \ra \sim \kr{F}_{c^{\scriptscriptstyle{(1)}}}(\ora{\D^{\scriptscriptstyle{(1)}}},\Delta^{\scriptscriptstyle{(1)}}_{n'+s/2}|q))\cdot \kr{F}_{c^{\scriptscriptstyle{(1)}}}(\ora{\D^{\scriptscriptstyle{(2)}}},\Delta^{\scriptscriptstyle{(2)}}_{n'+s/2}|q)),
\]
where the intermediate weight $(\Delta^{\scriptscriptstyle{(1)}}_{n'+s/2},\Delta^{\scriptscriptstyle{(2)}}_{n'+s/2})$ is the weight of the vector $|P',n'+s/2\rangle$, and the external weights are equal to
\begin{equation*}
\begin{aligned}
\ora{\D^{\scriptscriptstyle{(1)}}}=\Bigl(\Delta(P'^{\scriptscriptstyle{(1)}}_{n'},b^{\scriptscriptstyle{(1)}}),\Delta_{1,2}(b^{\scriptscriptstyle{(1)}}),\Delta(\alpha^{\scriptscriptstyle{(1)}}-Q^{\scriptscriptstyle{(1)}}/2,b^{\scriptscriptstyle{(1)}}),\Delta(P^{\scriptscriptstyle{(1)}}_{n},b^{\scriptscriptstyle{(1)}})\Bigr),\\
\ora{\D^{\scriptscriptstyle{(2)}}}=\Bigl(\Delta(P'^{\scriptscriptstyle{(2)}}_{n'},b^{\scriptscriptstyle{(2)}}),\Delta_{2,1}(b^{\scriptscriptstyle{(2)}}),\Delta(\alpha^{\scriptscriptstyle{(2)}}-Q^{\scriptscriptstyle{(2)}}/2,b^{\scriptscriptstyle{(2)}}),\Delta(P^{\scriptscriptstyle{(2)}}_{n},b^{\scriptscriptstyle{(2)}})\Bigr). 
\end{aligned}
\end{equation*}

It follows from equation \eqref{eq:phi_eq} that the conformal block with the degenerate vertex operator $\phi_{1,2}(q)$ satisfies a second order differential equation. This equation reduces to hypergeometric equation \cite{BPZ}. Therefore the conformal blocks written above are nonzero only for $s=\pm 1$  and become proportional to the ${}_2F_{1}$ hypergeometric function
\begin{equation}
\kr{F}^{(\eta)}_s(q)=q^{a^{(\eta)}_s}(1-q)^{d^{(\eta)}} {}_2F_1(A^{(\eta)}_s,B^{(\eta)}_s|C^{(\eta)}_s|q), \quad \eta=1,2, \label{eq:krF}
\end{equation}
where
\begin{equation}\label{eq:hyper}
\begin{aligned}
a_s^{\scriptscriptstyle{(1)}}&=\Delta(P'^{\scriptscriptstyle{(1)}}_{n'+s/2},b^{\scriptscriptstyle{(1)}})-\D_{1,2}(b^{\scriptscriptstyle{(1)}})-\D(P'^{\scriptscriptstyle{(1)}}_{n'},b^{\scriptscriptstyle{(1)}}), 
\\ 
d^{\scriptscriptstyle{(1)}}&=\D(\alpha^{\scriptscriptstyle{(1)}}-Q^{\scriptscriptstyle{(1)}}/2-\alpha_{1,2}^{\scriptscriptstyle{(1)}},b^{\scriptscriptstyle{(1)}})-\D(\alpha^{\scriptscriptstyle{(1)}}-Q^{\scriptscriptstyle{(1)}}/2,b^{\scriptscriptstyle{(1)}})-\D_{1,2}(b^{\scriptscriptstyle{(1)}})
\\ 
A_s^{\scriptscriptstyle{(1)}}&=1/2+2\alpha_{1,2}^{\scriptscriptstyle{(1)}}\left(sP'^{\scriptscriptstyle{(1)}}-\alpha^{\scriptscriptstyle{(1)}}-P^{\scriptscriptstyle{(1)}}+Q^{\scriptscriptstyle{(1)}}/2\right)
\\
B_s^{\scriptscriptstyle{(1)}}&=1/2+2\alpha_{1,2}^{\scriptscriptstyle{(1)}}\left(sP'^{\scriptscriptstyle{(1)}}-\alpha^{\scriptscriptstyle{(1)}}+P^{\scriptscriptstyle{(1)}}+Q^{\scriptscriptstyle{(1)}}/2\right)
\\
C_s^{\scriptscriptstyle{(1)}}&=1+4s\alpha_{1,2}^{\scriptscriptstyle{(1)}}P'^{\scriptscriptstyle{(1)}},
\end{aligned}
\end{equation}
and similarly for $a_s^{\scriptscriptstyle{(2)}}, d^{\scriptscriptstyle{(2)}}, A_s^{\scriptscriptstyle{(2)}}, B_s^{\scriptscriptstyle{(2)}}, C_s^{\scriptscriptstyle{(2)}}$ with the replacement $(1) \leftrightarrow (2)$ in the superscript and $\alpha_{1,2}^{\scriptscriptstyle{(1)}} \leftrightarrow \alpha_{2,1}^{\scriptscriptstyle{(2)}}$. Therefore we can write
\begin{equation}
\la P,n|\overline{\Phi}_{\alpha}(1)f(q)|P',n'\ra=
\sum_{s=\pm 1}\la P',n'+s/2|f(1)|P',n'\ra \cdot \la P,n|\overline{\Phi}_{\alpha}(1)|P',n'+s/2\ra \cdot \kr{F}^{\scriptscriptstyle{(1)}}_s(q) \kr{F}^{\scriptscriptstyle{(2)}}_s(q)\label{eq:4point:temp}
\end{equation}

The first two factors in the sum appear due to the chosen above normalization of the conformal block $\kr{F}(q)=q^{\Delta}\Bigl(1+q(\dots)\Bigr)$. By the definition $l_{n\, n'+s/2}(P,\alpha,P')=\la P,n|\overline{\Phi}_{\alpha}(1)|P',n'+s/2\ra$. Using the definition of the vectors $|P,n'\ra$ \eqref{eq:|P,n} and the normalization $\la P,n'|P,n'\ra=1$ we have
\begin{equation}
\begin{aligned}
\la P',n'+1/2|f(1)|P',n'\ra=\la P',n'+1/2|f_{-\frac{4(n'+1/2)-1}{2}}|P',n'\ra= \Omega_{n'+1/2}(P')/\Omega_{n'}(P')\\
\la P',n'-1/2|f(1)|P',n'\ra=\la P',n'-1/2|f_{\frac{4n'-1}{2}}|P',n'\ra= \Omega_{n'}(P')/\Omega_{n'-1/2}(P')
\end{aligned}\label{eq:f=Omeg}
\end{equation}
Here we used that $n'>0$. Substituting these expressions into \eqref{eq:4point:temp} we get
\begin{equation}
\la P,n|\Phi_{\alpha}(1)f(q)|P',n'\ra=
\sum_{s=\pm 1}\left(l_{n\, n'+s/2}(P,\alpha,P')\left(\frac{\Omega_{n'+s/2}(P')}{\Omega_{n'}(P')}\right)^s   \kr{F}^{\scriptscriptstyle{(1)}}_s(q) \kr{F}^{\scriptscriptstyle{(2)}}_s(q)\right) \label{degencalc}
\end{equation}

\noindent \textbf{Step 2.} Now we want to consider  expression \eqref{degencalc} at the region $|q|\gg 1$. 

For the left hand side we can use 
 $[\overline{\Phi}_{\alpha}(1),f(q)]=0$ and write
\[
\la P,n|\overline{\Phi}_{\alpha}(1)f(q)|P',n'\ra=\la P,n|f(q)\overline{\Phi}_{\alpha}(1)|P',n'\ra=
\la P',n'|\overline{\Phi}^{+}_{\alpha}(1)f^+(q)|P,n\ra.
\]
It follows from \eqref{eq:conj:operators} that $f(q)^{+}=-1/q f(1/q)$. Conjugation of $\Phi_\alpha$ was given in \eqref{eq:conj.vertex}. Therefore we have
\begin{equation}
\la P,n|\overline{\Phi}_{\alpha}(1)f(q)|P',n'\ra=\frac{(-1)^{2(n+n')}}q\la P',n'|\overline{\Phi}_{\alpha}(1)f(1/q)|P,n\ra
\label{crosssymrel}
\end{equation}
We can substitute \eqref{degencalc} to the right hand side of \eqref{crosssymrel} and see that $\la P,n|\overline{\Phi}_{\alpha}(1)f(q)|P',n'\ra$ is a linear combination of $\kr{F}^{\scriptscriptstyle{(1)}}_1(1/q) \kr{F}^{\scriptscriptstyle{(2)}}_1(1/q)$ and $\kr{F}^{\scriptscriptstyle{(1)}}_{-1}(1/q) \kr{F}^{\scriptscriptstyle{(2)}}_{-1}(1/q)$.

On the other hand the hypergeometric functions on $q$ and $1/q$ are connected by 
\begin{equation}
\begin{aligned}
F(A,B,C|q)=\frac{\Gamma(C)\Gamma(B-A)}{\Gamma(B)\Gamma(C-A)} (-q)^{-A}F(A,1-C+A,1-B+A,1/q)+\\+\frac{\Gamma(C)\Gamma(A-B)}{\Gamma(A)\Gamma(C-B)}(-q)^{-B}F(B,1-C+B,1-A+B,1/q). 
\end{aligned}\label{eq:conthyper}
\end{equation}
Therefore we have the relation for functions $\kr{F}^{(\eta)}_s(q)$ ($\eta=1,2$) defined as \eqref{eq:krF}
\begin{equation}
\kr{F}^{(\eta)}_s(q)=q^{-2\D_2^{(\eta)}}\sum\nolimits_{t=\pm 1}B_{st}^{(\eta)}\kr{F}^{(\eta)}_t(1/q) \label{ancont},
\end{equation}
The transformation matrix $B^{(\eta)}_{st}$ do not depend on $q$. 

We substitute \eqref{ancont} into \eqref{degencalc} and get the linear combination of $\kr{F}^{\scriptscriptstyle{(1)}}_{s}\kr{F}^{\scriptscriptstyle{(2)}}_{s'}$, for $s,s'=\pm 1$. But it was proven above that only the terms $s=s'$ can appear. Therefore the coefficient of the term $\kr{F}^{\scriptscriptstyle{(1)}}_{1}\kr{F}^{\scriptscriptstyle{(2)}}_{-1}$ should vanish. This is equivalent to the equation (for $n'>0$)
\begin{equation}
\frac{l_{n\, n'+1/2}(P,\alpha,P')}{l_{n\, n'-1/2}(P,\alpha,P')}=-\frac{ \Omega^2_{n'}(P')}{ \Omega_{n'+1/2}(P')\Omega_{n'-1/2}(P')} \frac{B_{-+}^{\scriptscriptstyle{(1)}}B_{--}^{\scriptscriptstyle{(2)}}}{B_{++}^{\scriptscriptstyle{(1)}} B_{+-}^{\scriptscriptstyle{(2)}}}. \label{eq:ln+1n}
\end{equation}

\noindent \textbf{Step 3.} Due to the relation between $\Phi(q)$ and $\Phi^+(1/q)$ we have
\begin{equation}
l_{n n'}(P,\alpha,P')=(-1)^{2(n+n')}l_{n' n}(P',\alpha,P).\label{eq:l=l:conj}
\end{equation}
Therefore we can find ${l_{n+\frac12\; n'}(P,\alpha,P')}/{l_{n-\frac12\; n'}(P,\alpha,P')}$ from \eqref{eq:ln+1n}. Using these relations we reduce $l_{n n'}$ to $l_{\{n\} \{n'\}}$, where $\{n\}$ denotes fractional part of $n$.

Due to \eqref{eq:l=l:conj} it is enough to consider only the case $n\geq n'$. In this case we use \eqref{eq:conthyper}, \eqref{ancont} and rewrite the ratio of the gamma functions as
\begin{equation*}
\begin{aligned}
\frac{B_{-+}^{\scriptscriptstyle{(1)}}B_{--}^{\scriptscriptstyle{(2)}}}{B_{++}^{\scriptscriptstyle{(1)}} B_{+-}^{\scriptscriptstyle{(2)}}}=
\frac{\prod\limits_{i,j\geq 1,\;i+j=2n'+2n+1}(\alpha-P-P'-ib-jb^{-1})\prod\limits_{i,j\geq 0,\;i+j=2n'-2n-1}(\alpha-P+P'+ib+jb^{-1})}
{\prod\limits_{i\geq 0, j \geq 1\;i+j=4n'}(2P'+ib+j b^{-1})}\times
\\ 
\times
\frac{\prod\limits_{i,j\geq 1,\;i+j=2n'-2n+1}(\alpha+P-P'-ib-jb^{-1})\prod\limits_{i,j\geq 0,\; i+j=2n+2n'-1}(\alpha+P+P'+ib+jb^{-1})}
{\prod\limits_{i\geq 1,j\geq 0,\;i+j=4n'}^{4n'-1}(2P'+ib+jb^{-1})}
\end{aligned}
\end{equation*}
%
%
Using this expression we get
\begin{multline}
l_{n n'}(P,\alpha,P')=
(-1)^{\lfloor n \rfloor+ \lfloor n' \rfloor}l_{\{n\} \{n'\}}(P,\alpha,P') \cdot\prod\limits_{i'=\{n'\}}^{n'-1}\frac{\Omega^2_{i'+1/2}(P')}{\Omega_{i'}(P')\Omega_{i'+1}(P')} \prod\limits_{i=\{n\}}^{n-1}\frac{\Omega^2_{i+1/2}(P)}{\Omega_{i}(P)\Omega_{i+1}(P)}\cdot 
\\ 
\cdot\frac{\prod_{\epsilon=\pm 1}\prod\nolimits_{\substack{i,j\geq1,\; i+j\equiv 2(n+n')\bmod 2\\
2+2(\{n\}+\{n'\} \leq i+j\leq 2(n'+\epsilon n))}}(\alpha-\epsilon P-P'-ib-jb^{-1}) (\alpha+\epsilon P+P'-Q+ib+jb^{-1}) }{\prod\limits_{\substack{i'>0,j'\geq1,\;i'+j'\leq 4n'-2\\i'+j'\equiv (2-4\{n'\})\bmod 4}}(2P'+j'b^{-1}+i'b)(2P'+i'b^{-1}+j'b) \prod\limits_{\substack{i>0,j\geq1,\;i+j\leq 4n-2\\i+j\equiv (2-4\{n\})\bmod 4}}(2P+jb^{-1}+ib))(2P+ib^{-1}+jb))}. \label{tempcal}
\end{multline}
The values of $l_{\{n\} \{n'\}}(P,\alpha,P')$ can be calculated explicitly using the expressions
\[
|P,0\ra=\overline{|P\ra},\qquad  |P,1/2\ra=\Omega_{1/2}(P)\left(f_{-1/2}+\frac{1}{Q/2+P}G_{-1/2}\right)\overline{|P\ra},\quad \Omega^2_{1/2}(P)=-\frac{Q/2+P}{2P}. 
\]
The answer reads
\begin{equation*}
\begin{aligned}
l_{0 0}(P,\alpha,P')=1, \qquad l_{\frac12 \frac12}(P,\alpha,P')=\frac{(Q+P+P'-\alpha)(P+P'+\alpha)}{\sqrt{4PP'(Q+2P)(Q+2P')}}
,\\
l_{0 \frac12}(P,\alpha,P')=\frac{-i}{\sqrt{(Q+2P')P'}}, \qquad l_{\frac12 0}(P,\alpha,P')=\frac{i}{\sqrt{(Q+2P)P}}.
\end{aligned}
\end{equation*}

\noindent \textbf{Step 4.} In order to finish the proof we should calculate the coefficients $\Omega_n(P)$. Below we used the generic values for $\alpha$. But the matrix elements \eqref{tempcal} are defined in algebraically and this formula should hold for any $\alpha$.

If we put $\alpha=0$ and then $P=P'$ then the operator $\overline{\Phi}_\alpha$ defined by \eqref{prim} and normalization is the identity operator. Therefore $l_{n n}(P, 0, P)=\la P,n|P,n\ra=1$. Substituting this into \eqref{tempcal} we have for ($n\geq 1/2$)
\begin{equation}
\begin{aligned}
\frac{\Omega^2_{n+1/2}(P)\Omega^2_{n-1/2}(P)}{\Omega^4_{n}(P)}=\frac{\prod\limits_{i,j\geq 0,\; i+j=4n-2}(2P+ib+jb^{-1}) \prod\limits_{i,j\geq 0, i+j=4n}(2P+Q+ib+jb^{-1})}{\prod\limits_{i\geq 0, j\geq 1\;i+j=4n}(2P+ib+jb^{-1})(2P+ib^{-1}+jb)}. \label{Omrel}
\end{aligned}
\end{equation}
Using the initial date $\Omega_{0}(P)=1$, $\Omega^2_{1/2}(P)=-\frac{Q/2+P}{2P}$ we get the answer  \eqref{eq:Omega}. Substituting this to \eqref{tempcal} we get the answer \eqref{melemans}.

\end{proof}

Several remarks are in order.

\begin{Remark}
The main point of the proof was the equation \eqref{eq:ln+1n} which follows from the vanishing of the coefficient of the $\kr{F}^{\scriptscriptstyle{(1)}}_{1}\kr{F}^{\scriptscriptstyle{(2)}}_{-1}$ term. A similar vanishing of  the $\kr{F}^{\scriptscriptstyle{(1)}}_{-1}\kr{F}^{\scriptscriptstyle{(2)}}_{1}$ term imposes the relation, which differs from \eqref{eq:ln+1n} by the replacement $b\leftrightarrow b^{-1}$. But the final answers \eqref{melemans} and \eqref{eq:Omega} are symmetric under the $b\leftrightarrow b^{-1}$ replacement. Therefore this new relation does ont impose new constraint. 
\end{Remark}

\begin{Remark} \label{rem:Omega}
In this paper we fix $\Omega_n(P)$ by the relation $\la P,n|P,n\ra=1$. But the expression for matrix elements \eqref{melemans} should be valid for any $\Omega_n(P)$ since they appear only as factors.

In \cite{BBFLT} another normalization was used\footnote{More precisely this is a corrected formula, it looks like there is a misprint in \cite{BBFLT}.}
\[
\td{\Omega}_n=2^{-n}s_{\textrm{even}}(2P+Q,2n).
\]
Substituting this into \eqref{melemans} we get the formula for the matrix elements proposed in \cite{BBFLT} (up to a sign factor). 
\end{Remark}

As was explained in the end of the previous subsection formula \eqref{melemans} gives an explicit expression for $l_n^{21}$. Substituting $\Omega_n$ from \eqref{eq:Omega} and putting $n'=0$, we get
\begin{equation}
\begin{aligned}
\!\!\!\!\!l_n^{21}(P,b|\D^{\NS}_1,\D^{\NS}_2)=\frac{(-1)^n}{\sqrt{s_{\textrm{even}}(2P,2n) s_{\textrm{even}}(2P+Q,2n)}} \times
  \begin{cases}
  \prod_{\epsilon,\epsilon'=\pm}
  s_{\textrm{even}}(P_2+\epsilon P_1+\epsilon' P+\frac{Q}2,\epsilon' n), \;\; n+n' \in \mathbb{Z}\\
  \prod_{\epsilon,\epsilon'=\pm}
  s_{\textrm{odd}}(P_2+\epsilon P_1+\epsilon' P+\frac{Q}2,\epsilon' n),
  \;\; n+n'\in\mathbb{Z}+\frac12
  \end{cases} \label{l_n}
\end{aligned}
\end{equation}
Then we have from \eqref{lirrlim}
\begin{equation}
l_n(P,b)=\frac{(-1)^{n}2^{2n^2}(\beta^{\scriptscriptstyle{(1)}})^{-\D^{\scriptscriptstyle{(1)}}_n/2}(\beta^{\scriptscriptstyle{(2)}})^{-\D^{\scriptscriptstyle{(2)}}_n/2}}{\sqrt{s_{\textrm{even}}(2P,2n) s_{\textrm{even}}(2P+Q,2n)}} \label{l_n_irr}
\end{equation}

\begin{Remark}
There exist another method to find $\Omega_n(P)$. Consider the function defined in the region $|q|<1$ by the formula
\[
H_n(q)=\la P,n|f(1)f(q)|P,n\ra.
\] 
We see that
\begin{equation}
H_n(q)= \frac{1}{1-q}+\sum_{s>0}\la P,n|f_{-s}f_{s}|P,n\ra (q^{-s-1/2}- q^{s-1/2}) \label{eq:Hn}
\end{equation}
and the sum is actually finite for any $n$. Therefore $H_n(q)$ is a sum of Laurent polynomial and $\frac{1}{1-q}$. Therefore $H_n(q)$ can be analytically continued on the $ \mathbb{CP}^1\setminus\{0,1,\infty\}$. In the $|q|>1$ region we have $H_n(q)=-\la P,n|f(q)f(1)|P,n\ra$ i.e. $H_n(q)$ is a radial ordered conformal block.  It follows from \eqref{eq:Hn} that 
\[
H_n(q)=-q^{-1} H_n(1/q)
\]
similarly to \eqref{crosssymrel}. 

Using the second order differential equations we can write the function $H_n(q)$ as a sum of two products of hypergeometric functions similarly to \eqref{degencalc}. The parameters of the hypergeometric functions are given by \eqref{eq:hyper} for $P'=P$, $n'=n$, $\alpha^{\scriptscriptstyle{(1)}}=\alpha^{\scriptscriptstyle{(1)}}_{1,2}$, $\alpha^{\scriptscriptstyle{(2)}}=\alpha^{\scriptscriptstyle{(2)}}_{2,1}$. We have $\la P,n|f(1)|P,n'\ra$ instead of $l_{n n'}$ in the coefficients of analogue of \eqref{degencalc}. Recall that $\la P,n|f(1)|P,n'\ra$ are given in terms of $\Omega_n$ (see \eqref{eq:f=Omeg}). Therefore the vanishing of the coefficient of $\kr{F}^{\scriptscriptstyle{(1)}}_{1}\kr{F}^{\scriptscriptstyle{(2)}}_{-1}$ term gives an equation on $\Omega$ in terms of the transformation matrix $B^{ (\eta)}_{st}$
\[
\frac{\Omega^2_{n+1/2} \Omega^2_{n-1/2}}{\Omega^4_{n}}=-\frac{{B}_{-+}^{\scriptscriptstyle{(1)}}{B}_{--}^{\scriptscriptstyle{(2)}}}{{B}_{++}^{\scriptscriptstyle{(1)}} {B}_{+-}^{\scriptscriptstyle{(2)}}}.
\]
As expected, this relation coincides with \eqref{Omrel}.
\end{Remark}

\section{Bilinear relations on conformal blocks}
\label{sec:bilin}
\subsection{Painlev\'e equations and isomonodromic problem} \label{ssec:Painleve}

Painlev\'e equations are second-order differential equations with no movable brunching points except poles.  We recall several facts about them following \cite{GIL1302}, \cite{myCA}.

The Painlev\'e VI equation has the form
 \begin{align*}
 &\,\frac{d^2q}{dt^2}=\frac{1}{2}\left(\frac{1}{q}+\frac{1}{q-1}+\frac{1}{q-t}\right)\left(\frac{dq}{dt}\right)^2-
 \left(\frac{1}{t}+\frac{1}{t-1}+\frac{1}{q-t}\right)\frac{dq}{dt}+\\
 \nonumber
 &\,+\frac{2q(q-1)(q-t)}{t^2(t-1)^2}\left(\left(\theta_{\infty}-\text{\small $\frac12$}\right)^2-\frac{\theta_0^2t}{q^2}+
 \frac{\theta_1^2(t-1)}{(q-1)^2}-
 \frac{\left(\theta_t^2-\frac14\right)t(t-1)}{(q-t)^2}\right),
\end{align*}
Here $\theta_0,\theta_1,\theta_t,\theta_\infty$ are the parameters of the equation. All other Painlev\'e equations can be obtained from Painlev\'e $\mathrm{VI}$ by a confluence
\begin{center}
 \begin{tikzpicture}[node distance=1.8cm, auto]
 \node (P6) {P$\mathrm{VI}$};
 \node (P5) [right of=P6] {P$\mathrm{V}$};
 \node (P31) [right of=P5] {P$\mathrm{III_1}$};
 \node (P32) [right of=P31] {P$\mathrm{III_2}$};
 \node (P33) [right of=P32] {P$\mathrm{III_3}$};
 \node (P4) [below of=P5] {P$\mathrm{IV}$};
 \node (P2) [below of=P31] {P$\mathrm{II}$};
 \node (P1) [below of=P32] {P$\mathrm{I}$};
 \draw[->] (P6) to node {} (P5);
  \draw[->] (P5) to node {} (P31);
   \draw[->] (P31) to node {} (P32);
    \draw[->] (P32) to node {} (P33);
     \draw[->] (P5) to node {} (P4);
      \draw[->] (P4) to node {} (P2);
       \draw[->] (P2) to node {} (P1);
        \draw[->] (P31) to node {} (P2);
         \draw[->] (P32) to node {} (P1);
 \end{tikzpicture}
 \end{center}
The Painlev\'e III$'_3$ equation has the form
\begin{equation}\label{piiispr}
\frac{d^2q}{dt^2}=\frac{1}{q}\left(\frac{dq}{dt}\right)^2 -
 \frac{1}{t}\frac{dq}{dt}\,+\frac{2q^2}{t^2}-\frac{2}{t},
 \end{equation}
Remark that the Painlev\'e III$'_3$ equation differs from standard Painlev\'e III$_3$ equation by the change of variables $t_{_{\mathrm{III'}}}=t_{_{\mathrm{III}}}^2$,
 $q_{_{\mathrm{III'}}}=t_{_{\mathrm{III}}}q_{_{\mathrm{III}}}$.

We now proceed to the Hamiltonian (or $\zeta$) form of Painlev\'e VI and III$'_3$ and then to the $\tau$ form. For simplicity we present all formulae  for the Painlev\'e III$'_3$ equation and omit some analogous calculations for the Painlev\'e~VI equation. 

The Painlev\'e equations can be rewritten as non-autonomous Hamiltonian systems.
It means that they can be obtained by eliminating an auxiliary momentum $p(t)$ from the equations
\[
\frac{dq}{dt}=\frac{\partial H_{_{\mathrm{J}}}}{\partial p},\qquad \frac{dp}{dt}=-\frac{\partial H_{_{\mathrm{J}}}}{\partial q},\qquad\qquad\mathrm{J}=\mathrm{VI},\mathrm{III}'_{3} \textrm{ and others.}
\]
The corresponding Hamiltonians are given by the expressions
\begin{align}
 t(t-1)H_{_{\mathrm{VI}}}=
 &q\left(q-1\right)\left(q-t\right)p\left(p-\frac{2\theta_0}{q}-\frac{2\theta_1}{q-1}-
 \frac{2\theta_t-1}{q-t}\right)+
 \label{hamp6}\\ &\qquad\qquad\qquad\qquad\qquad \notag
 +\left(\theta_0+\theta_t+\theta_1+\theta_{\infty}\right)
 \left(\theta_0+\theta_t+\theta_1-\theta_{\infty}-1\right)q,  \\
 tH_{_{\mathrm{III}_3'}}=&p^2q^2-q-\frac{t}{q}  \label{hamp33}.
\end{align}
It is convenient to pass from the Hamiltonians to closely related functions $\zeta(t)$ by the formulae
\begin{equation*}
\begin{aligned}
\zeta_{\mathrm{VI}}(t)&=t(t-1)H_{\mathrm{VI}}(t)-q(q-1)p+(\theta_0+\theta_t+\theta_1+\theta_{\infty})q-2\theta_0\theta_t-2\theta_0^2 t-2\theta_0\theta_1 t \\
\zeta_{\mathrm{III'_3}}(t)&=tH_{\mathrm{III'_3}}(t).
\end{aligned}
\end{equation*}
Remark that if we know the functions $\zeta(t)$ on trajectories of motion then we can find $q(t)$ and $p(t)$. For the Painlev\'e III$'_3$ equation we have
\[
q(t)=-\frac{1}{\zeta'(t)}, \quad p(t)=t\zeta''(t)/2. 
\]
See \cite{GIL1302} for the analogous expression of $q(t)$ in terms of $\zeta(t)$,$\zeta'(t)$,$\zeta''(t)$ for the Painlev\'e VI equation. Substituting these expressions to \eqref{hamp33} we get Hamiltonian (or $\zeta$) form of Painlev\'e III$'_3$ equation
\begin{equation}
(t\zeta''(t))^2=4(\zeta'(t))^2(\zeta(t)-t\zeta'(t))-4\zeta'(t). \label{sg3}
\end{equation}
The Painlev\'e VI equation in $\zeta$ form reads
\begin{equation}
(t(t-1)\zeta''(t))^2=-2 \det
\left( \begin{array}{ccc}
2\Delta_0 & t\zeta'(t)-\zeta(t) & \zeta'(t)+\Delta_0+\Delta_t+\Delta_1-\Delta_\infty\\
t\zeta'(t)-\zeta(t) & 2\Delta_t & (t-1)\zeta'(t)-\zeta(t) \\
\zeta'(t)+\Delta_0+\Delta_t+\Delta_1-\Delta_\infty & (t-1)\zeta'(t)-\zeta(t) & 2\Delta_1
\end{array} \right) \label{sg6}
\end{equation}
Here $\Delta_{\nu}=\theta_{\nu}^2$, for $\nu=0,1,t,\infty$. 

Now let us differentiate \eqref{sg3}, \eqref{sg6} and divide the result by $\zeta''(t)$. Substitute $\zeta(t)=t(t-1)\frac{d\log\tau(t)}{dt}$ in the Painlev\'e VI case and 
$\zeta(t)=t\frac{d \log\tau(t)}{dt}$ in the Painlev\'e III$'_3$ case. We obtain bilinear equations on the $\tau$~functions. It is convenient to write these equations by use of Hirota differential operators $D^k_{[x]}$. In our paper we use only Hirota derivatives with respect to the logarithm of a variable. These operators are defined by the formula
\begin{equation} \label{eq:Hirota:def}
f(e^{\alpha}t)g(e^{-\alpha}t)=\sum\limits_{k=0}^{\infty}D^{k}_{[\log t]}(f(t),g(t))\frac{\alpha^k}{k!}.
\end{equation}
The first examples of Hirota derivatives are
\[
D^0_{[\log t]}(f(t),g(t))=f(t)g(t),\qquad D^1_{[\log t]}(f(t),g(t))=tf'(t)g(t)-f(t)tg'(t). 
\]
Then, the $\tau$ form of the Painlev\'e III$'_3$ equation can be written as follows
\begin{equation}
D^{III}(\tau(t),\tau(t))=0,\quad \text{where} \quad D^{III}=
\frac12 D^{4}_{[\log t]}-t\frac{d}{dt}D^{2}_{[\log t]}+\frac12D^{2}_{[\log t]}+2tD^{0}_{[\log t]} \label{tau3}
\end{equation}
This form was found in \cite{Okamoto:2006}. For the Painlev\'e VI case we use $\tilde{\tau}(t)=t^{\Delta_0+\Delta_t}\tau(t)$ and rewrite the equation as $D^{VI}(\tilde{\tau}(t),\tilde{\tau}(t))=0$, where
\begin{equation}
\begin{aligned}
D^{VI}=&-\frac12(1-t)^3{D}_{[\log t]}^{4}+(1-t)^2(1+t) \left(t\frac{d}{dt}\right){D}_{[\log t]}^{2}+\\
&+(1-t)\Bigl(2 t(\D_t+\D_1)-2(1-t)t(\D_0+\D_\infty)-\frac12 (1-t+t^2)\Bigr){D}_{[\log t]}^{2}-\\
&-\frac12t(1-t)\left(t\frac{d}{dt}\right)^2{D}_{[\log t]}^{0}+ t\Bigl((\D_0+\D_\infty)(1-t)-(\D_t+\D_1)(1+t)\Bigr)\left(t\frac{d}{dt}\right){D}_{[\log t]}^{0}+\\
&+2t\Bigl((\D_0-\D_t)(\D_1-\D_\infty)+t(\D_0+\D_t)(\D_1+\D_\infty)\Bigr){D}_{[\log t]}^{0}
\end{aligned}.\notag
\end{equation}

\bigskip

Now let us review some facts about the isomonodromic deformations of linear systems on $\mathbb{CP}^1$(following \cite{GIL1207}). In the simplest non-trivial case this problem leads to the Painlev\'e VI equation.

We start from a linear system of rank $N$ with $n$ regular singularities $a=\{a_1,\ldots,a_n\}$ on $\mathbb{CP}^1$
\begin{equation}
 \label{dphiz}
 \partial_z\Phi=\mathcal{A}(z)\Phi, \qquad \mathcal{A}(z)=\sum_{\nu=1}^n\frac{\mathcal{A}_{\nu}}{z-a_{\nu}},
\end{equation}
 where $\{\kr{A}_{\nu}\}$ are $\mathfrak{sl}(N,\mathbb{C})$ constant matrices.
 
 We made some assumptions. We assume the constraint $\sum\nolimits_{\nu=1}^n\mathcal{A}_{\nu}=0$, which is equivalent to  the absence of singularity at $\infty$. We assume that $\mathcal{A}_{\nu}$ are diagonalizable so that $\mathcal{A}_{\nu}=\mathcal{G}_{\nu}\mathcal{T}_{\nu}\mathcal{G}_{\nu}^{-1}$
 with some $\mathcal{T}_{\nu}=\mathrm{diag}\left\{\lambda_{\nu,1},\ldots,\lambda_{\nu,N}\right\}$. And finally we assume that $\lmb_{\nu,j}-\lmb_{\nu,k}\notin\mathbb{Z}$ for $j\neq k$ (a non-resonance assumption). 

The fundamental solution is normalized by $\Phi(z_0)=\mathbf{1}_{N}$. Near the singular points, the fundamental solution has the following expansions
\[
 \Phi(z\rightarrow a_{\nu})=\mathcal{G}_{\nu}(z)\left(z-a_{\nu}\right)^{\mathcal{T}_{\nu}}\mathcal{C}_{\nu}.
\]
 Here $\mathcal{G}_{\nu}(z)$ is holomorphic and invertible in a neighborhood of $z=a_{\nu}$
 and satisfies $\mathcal{G}_{\nu}(a_{\nu})=\mathcal{G}_{\nu}$. The matrix $\mathcal{C}_{\nu}$ is independent of $z$ and is defined by the position of $z_0$. Counterclockwise continuation of $\Phi(z)$ around $a_{\nu}$ leads to a monodromy matrix $\mathcal{M}_{\nu}=\mathcal{C}_{\nu}^{-1}e^{2\pi i \mathcal{T}_{\nu}}\mathcal{C}_{\nu}$.

 Let us now vary the positions of singularities and $\mathcal{A}_{\nu}$'s in such way that the monodromy is preserved. A classical result translates this requirement into a system of PDEs
 \begin{align}
 \label{dphia}
 \partial_{a_{\nu}}\Phi=&\,-\frac{z_0-z\;\,}{z_0-a_{\nu}}\,\frac{\mathcal{A}_{\nu}}{z-a_{\nu}}\,\Phi,
 \end{align}
Schlesinger deformation equations are obtained as
 compatibility conditions of (\ref{dphiz}) and (\ref{dphia}). Explicitly,
\[
 \partial_{a_{\mu}}\mathcal{A}_{\nu}=\frac{z_0-a_{\nu}}{z_0-a_{\mu}}\,
 \frac{\left[\mathcal{A}_{\mu},\mathcal{A}_{\nu}\right]}{a_{\mu}-a_{\nu}},\quad \mu\neq\nu,\qquad\qquad
 \partial_{a_{\nu}}\mathcal{A}_{\nu}=-\sum_{\mu\neq\nu}
 \frac{\left[\mathcal{A}_{\mu},\mathcal{A}_{\nu}\right]}{a_{\mu}-a_{\nu}}.
\]
It follows from Schlesinger equations that the form $\sum_{\mu<\nu}\tr\,\mathcal{A}_{\mu}\mathcal{A}_{\nu} \;d\log\left(a_{\mu}-a_{\nu}\right)$ is closed. The isomonodromic $\tau$~function $\tau(a)$ is locally defined by
\begin{equation}
\label{taudef}
d\log\tau=\sum_{\mu<\nu}\tr\,\mathcal{A}_{\mu}\mathcal{A}_{\nu} \;d\log\left(a_{\mu}-a_{\nu}\right).
\end{equation}

The equation \eqref{dphiz} is invariant under fractional linear transformations, i.e. if $w=w(z)$, then transformed equation have the form:
\[
\partial_w\Phi=\sum_{\nu=1}^n \frac{\mathcal{A}_\nu}{z-w(a_{\nu})}
\]
It can be shown that the isomonodromic $\tau$~function of the system with $n$ singularities transforms under fractional linear transformations identically to $n$-point chiral correlation function of CFT primaries with conformal weights $\D_{\nu}=\frac{1}{2}\tr \kr{A}_{\nu}^2$. 
The expression of the Painlev\'e VI  $\tau$~function in terms of the Liouville conformal blocks $c=1$ from \cite{GIL1207} is a generalisation of this observation.

Now let $\mathcal{A}_\mu\in \mathfrak{sl}(2)$ and $n$ (number of singular points) equal to 4. We will see that in this case the isomonodromic $\tau$~function coincides with Painlev\'e $\tau$~function. We follow \cite{myCA} in the presentation.

Using the fractional-linear transformations we put singular points to $0,t,1,\infty$ (we now drop the constraint $\sum\nolimits_{\nu=1}^n\mathcal{A}_{\nu}=0$). We move the normalization point $z_0$ to $\infty$ choosing  an appropriate asymptotics of $\Phi(z)$ as normalization condition (or we can keep $z_0$ finite and the resulting equation remains the same since the $\tau$~function does not depend on $z_0$). Then the Schlesinger equations have the form
\begin{equation}
\partial_t \kr{A}_0=\frac{[\kr{A}_t,\kr{A}_0]}{t}, \quad \partial_t \kr{A}_1=\frac{[\kr{A}_t,\kr{A}_1]}{t-1}, \quad \partial_t \kr{A}_t=-\frac{[\kr{A}_t,\kr{A}_0]}{t}-\frac{[\kr{A}_t,\kr{A}_1]}{t-1}. \label{PShl}
\end{equation}
We introduce $\zeta(t)=t(t-1)\frac{d \log \tau(t)}{dt}$. From \eqref{taudef} we get the relation
\[
\zeta(t)=(t-1)\tr\kr{A}_t\kr{A}_0+t\,\tr\kr{A}_t\kr{A}_1. 
\]
Differentiating and using \eqref{PShl} one finds
\begin{equation}
\zeta'(t)=\tr\kr{A}_t\kr{A}_0+\tr\kr{A}_t\kr{A}_1, \qquad
\zeta''(t)=\frac{\tr(\kr{A}_0[\kr{A}_t,\kr{A}_1])}{t(1-t)}. \label{deriv}
\end{equation}

Now one can use an identity valid for any triple of matrices $\kr{A}_0,\mathcal{A}_t,\mathcal{A}_1 \in \mathfrak{sl}(2)$
\[
\tr([\kr{A}_0,\kr{A}_t]\kr{A}_1)^2=-2 \det
\left( \begin{array}{ccc}
\tr\kr{A}_0^2 & \tr\kr{A}_0\kr{A}_t & \tr\kr{A}_0\kr{A}_1\\
\tr\kr{A}_t\kr{A}_0 & \tr\kr{A}_t^2 & \tr\kr{A}_t\kr{A}_1 \\
\tr\kr{A}_1\kr{A}_0 & \tr\kr{A}_1\kr{A}_t & \tr\kr{A}_1^2
\end{array} \right).
\]
This identity is equivalent to the well known formula for a triple product of vectors in $\mathbb{R}^3$. Substituting \eqref{deriv} and $\D_{\nu}=\frac12 \tr{\kr{A_{\nu}}^2}$ we get a differential equation on $\zeta(t)$ which coincides with \eqref{sg6}. Therefore the Painlev\'e VI $\tau$~function and the $\tau$~function of given case of isomonodromic problem coincide.

\subsection{Proof of the Painlev\'e III$'_3$ $\tau$~function conjecture}
\label{ssec:PIII}

In this subsection we prove Theorem \ref{thm:1}. Recall that we want to prove that the $\tau$~function defined by the expression
\begin{equation}\label{eq:tau:P3}
\tau(t)=\sum_{n\in\mathbb{Z}}s^n C(\sg+n) \kr{F}((\sg+n)^2|t),
\end{equation}
satisfy $D^{III}(\tau(t),\tau(t))=0$, see \eqref{tau3}. Here $\kr{F}(\sg^2|t)=\kr{F}_1(\sg^2|t)$ denotes the irregular limit of conformal block defined in \eqref{eq:Whitcom} for the central charge $c=1$. The coefficients $C(\sg)$ are defined by the formula
\[
C(\sg)=\frac{1}{\G(1-2\sg)\G(1+2\sg)},
\]
where $\G(z)$ is the Barnes $\G$ function. Of all properties of this function we will use only a recurrence relation: $\G(z+1)=\Gamma(z)\G(z)$. The parameters $s$ and $\sg$ in \eqref{eq:tau:P3} are constants of integration of the equation Painlev\'e $\mathrm{III'_3}$ \eqref{piiispr} (see also Remark \ref{rem:int const}).

\begin{proof}First we substitute the conjectural expression for $\tau$~function \eqref{eq:tau:P3} into \eqref{tau3} and collect terms with the same powers of $s$. The vanishing condition of the $s^m$ coefficient has the form
\[
\sum_{n\in\mathbb{Z}}\left( C(\sg+n+m)C(\sg-n) D^{III}\Bigl(\kr{F}((\sigma+n+m)^2|t),\kr{F}((\sigma-n)^2|t\Bigr)\right)=0.
\]
Clearly this $s^{m}$ term coincides with the $s^{m+2}$ term after the shift $\sigma \mapsto \sigma+1$. Therefore it is sufficient to prove the vanishing of $s^0$ and $s^1$ terms:
\begin{align}
\sum_{n\in\mathbb{Z}}\left(C(\sg+n)C(\sg-n) D^{III}\Bigl(\kr{F}((\sigma+n)^2|t),\kr{F}((\sigma-n)^2|t)\Bigr)\right)&=0, \label{s0}\\
\sum_{n\in\mathbb{Z}}\left(C(\sigma+n+1)C(\sg-n) D^{III}\Bigl(\kr{F}((\sigma+n+1)^2|t),\kr{F}((\sigma-n)^2|t)\Bigr)\right)&=0\label{s1}.
\end{align}
We prove these relations by use of the Whittaker vector decomposition proved in Proposition \ref{prop:Whitdecomp}. Taking the scalar square of \eqref{Whitdecomp} we have
\begin{equation} \label{blockdecomp}
\kr{F}_{c^{\NS}}(\Delta^{\NS}|q)=\sum\limits_{2n\in\mathbb{Z}}l_n^2(P,b) \kr{F}^{\scriptscriptstyle{(1)}}_{n}\kr{F}^{\scriptscriptstyle{(2)}}_{n},
\end{equation}
where
\begin{equation*}
\mathcal{F}^{\scriptscriptstyle{(1)}}_{n}=\mathcal{F}_{c^{\scriptscriptstyle{(1)}}}(\Delta_n^{\scriptscriptstyle{(1)}}|\beta^{\scriptscriptstyle{(1)}}q), \;\; \kr{F}^{\scriptscriptstyle{(2)}}_{n}=\mathcal{F}_{c^{\scriptscriptstyle{(2)}}}(\Delta_{n}^{\scriptscriptstyle{(2)}}|\beta^{\scriptscriptstyle{(2)}}q).
\end{equation*} 
We will use below the shorten notations $\kr{F}^{\scriptscriptstyle{(1)}}_{n}$,  $\kr{F}^{\scriptscriptstyle{(2)}}_{n}$. We want to prove relations that contain Hirota differential operators. Let us introduce the operator $H$
\begin{equation}
H=b L_{0}^{\scriptscriptstyle{(1)}}+b^{-1} L_{0}^{\scriptscriptstyle{(2)}},\label{eq:Hdefin}
\end{equation}
and define $\widehat{\kr{F}_{\NS}}$ and $\widehat{\kr{F}_k}$ by the formulae
\begin{equation}\label{eq:Fhat}
\widehat{\kr{F}_{\NS}}=\la 1\otimes W_{\NS}|e^{\alpha H}|1\otimes W_{\NS}\ra=\sum\limits_{k=0}^{\infty}\la 1\otimes W_{\NS}|H^{k}|1\otimes W_{\NS}\ra \frac{\alpha^k}{k!}=\sum\limits_{k=0}^{\infty}\widehat{\kr{F}_k} \frac{\alpha^k}{k!}
\end{equation}

We can calculate $\widehat{\kr{F}_k}$ using right hand side of \eqref{Whitdecomp}
\begin{multline}
\!\!\!\la 1\otimes W_{\NS}(q)|e^{H\alpha}|1\otimes W_{\NS}(q)\ra=\\= \sum\limits_{2n\in\mathbb{Z}}l_n^2(P,b)   \left\langle W^{\scriptscriptstyle{(1)}}_{n}(\beta^{\scriptscriptstyle{(1)}}q)\right|e^{\alpha b L_{0}^{\scriptscriptstyle{(1)}}}\left|W^{\scriptscriptstyle{(1)}}_{n}(\beta^{\scriptscriptstyle{(1)}}q)\right\rangle \left\langle W_{n}^{\scriptscriptstyle{(2)}}(\beta^{\scriptscriptstyle{(2)}}q)\right|e^{\alpha b^{-1} L_{0}^{\scriptscriptstyle{(2)}}}\left|W^{\scriptscriptstyle{(2)}}_{n}(\beta^{\scriptscriptstyle{(2)}}q)\right\rangle.\label{eq:FNS1}
\end{multline}
Generalized Hirota differential operators $D^{n}_{\epsilon_1,\epsilon_2[x]}$ are defined by 
\[
f(e^{\epsilon_1\alpha}q)g(e^{\epsilon_2\alpha}q)=\sum\limits_{n=0}^{\infty}D^{n}_{\epsilon_1,\epsilon_2[\log q]}(f(q),g(q))\frac{\alpha^n}{n!},
\]
where we take derivatives with respect to logarithm of variable as before. Since 
\[\Bigl\langle W^{(\eta)}_n(\beta^{(\eta)}q)\Bigr|e^{ \alpha\epsilon L_{0}^{(\eta)}}\Bigl|W^{(\eta)}_n(\beta^{(\eta)}q)\Bigr\rangle=\kr{F}^{(\eta)}_n((\beta^{(\eta)}q)e^{\alpha \epsilon}),\quad \eta=1,2.\]
we can rewrite \eqref{eq:FNS1} as
\begin{equation}
\label{eq:FhatHirota}
\widehat{\kr{F}_k}=\sum\limits_{2n\in\mathbb{Z}}\left(l_n^2(P,b)  D^{k}_{b,b^{-1}[\log q]}(\kr{F}^{\scriptscriptstyle{(1)}}_{n},\kr{F}^{\scriptscriptstyle{(2)}}_{n})\right).
\end{equation}

On the other hand we can calculate $\widehat{\kr{F}_k}$ using the left hand side of \eqref{Whitdecomp}. Using the explicit expressions~\eqref{Vir12} we can rewrite the operator $H$ in terms of the $\mathsf{F}\oplus\NSR$ generators
\begin{equation}
\label{eq:HNSR}
H=Q\sum\limits_{r\in\mathbb{Z}+1/2}r:f_{-r}f_{r}:-\sum\limits_{r\in\mathbb{Z}+1/2}f_{-r}G_r,
\end{equation}
We want to calculate $H^k|1\otimes W_{\NS}\ra$ and substitute this to \eqref{eq:Fhat}. We do this calculation for $k\leq 4$
\begin{align*}
H|1\otimes W_{\NS}\ra=&-q^{1/4}f_{-1/2}|1\ra\otimes|W_{\NS}\ra\\
H^2|1\otimes W_{\NS}\ra=&-q^{1/4}|1\ra\otimes G_{-1/2}|W_{\NS}\ra-Q q^{1/4}f_{-1/2}|1\ra\otimes|W_{\NS}\ra\\
H^3|1\otimes W_{\NS}\ra=&-Q^2 q^{1/4} f_{-1/2}|1\ra\otimes|W_{\NS}\ra-Q q^{1/4} |1\ra\otimes G_{-1/2}|W_{\NS}\ra\ny\\&+2q^{1/4}f_{-1/2}|1\ra\otimes L_0|W_{\NS}\ra+2q^{3/4}f_{-3/2}|1\ra\otimes |W_{\NS}\ra- q^{1/2}f_{-1/2}|1\ra\otimes G_{-1/2}|W_{\NS}\ra\\
H^4|1\otimes W_{\NS}\ra=&-Q^2 q^{1/4} |1\ra\otimes G_{-1/2}|W_{\NS}\ra+2 q^{1/4} |1\ra\otimes G_{-1/2} L_0|W_{\NS}\ra\ny\\
&-q^{1/2}|1\ra\otimes L_{-1}|W_{\NS}\ra+2q^{3/4}|1\ra\otimes G_{-3/2}|W_{\NS}\ra+\ldots,
\end{align*}
where $"\ldots"$ stands for terms involving $f$ in $H^4$. Then we have the relations
\begin{equation}
\begin{aligned}
\widehat{\kr{F}_0}=\kr{F}_{\NS},\qquad \widehat{\kr{F}_2}=-q^{1/2} \kr{F}_{\NS},\qquad
\widehat{\kr{F}_4}=q^{1/2}(2q\frac{d}{dq}\kr{F}_{\NS}-q^{1/2}\kr{F}_{\NS})-
Q^2 q^{1/2}\kr{F}_{\NS},
\end{aligned} \label{eq:hatF}
\end{equation}
where $\kr{F}_{\NS}$ denotes $\kr{F}_{c^{\NS}}(\Delta^{\NS}|q)=\la W_{\NS}|W_{\NS}\ra$ and we used the relation $ \la W_{\NS}|L_{0}|W_{\NS}\ra=q\frac{d}{dq}\kr{F}_{\NS}$. Using these formulas we derive the equation for $\widehat{\kr{F}_k}$
\begin{equation}
\widehat{\kr{F}_4}+2q\frac{d}{dq}\widehat{\kr{F}_2}-(1+Q^2)\widehat{\kr{F}_2}+q\widehat{\kr{F}_0}=0. \label{eq:eqFhat}
\end{equation}
Now we can use \eqref{eq:FhatHirota} and rewrite \eqref{eq:eqFhat} as a bilinear differential equation. Introduce the corresponding operator by the formula
\begin{equation}
D^{III}_b=
D^{4}_{b,b^{-1}[\log q]}+2q\frac{d}{dq}D^{2}_{b,b^{-1}[\log q]}-(1+Q^2)D^{2}_{b,b^{-1}[\log q]}+qD^{0}_{b,b^{-1}[\log q]}\label{d3} 
\end{equation}
We proved that
\begin{equation}
\sum\limits_{2n\in\mathbb{Z}}\left(l_n^2(P,b)\cdot D^{III}_b\Bigl(\kr{F}_{c^{\scriptscriptstyle{(1)}}}(\Delta_n^{\scriptscriptstyle{(1)}}|\beta^{\scriptscriptstyle{(1)}}q),\kr{F}_{c^{\scriptscriptstyle{(2)}}}(\Delta_{n}^{\scriptscriptstyle{(2)}}|\beta^{\scriptscriptstyle{(2)}}q)\Bigr)\right)=0 \label{bilin}
\end{equation}
This sum splits into two, which consist of integer and half integer $n$ ($q^{\Delta^{\NS}}$ times integer and half integer powers of $q$ correspondingly). Therefore we have
\begin{align}
\sum\limits_{n\in\mathbb{Z}}\left( l_n^2(P,b)\cdot D^{III}_b\Bigl(\kr{F}_{c^{\scriptscriptstyle{(1)}}}(\Delta_n^{\scriptscriptstyle{(1)}}|\beta^{\scriptscriptstyle{(1)}}q),\kr{F}_{c^{\scriptscriptstyle{(2)}}}(\Delta_{n}^{\scriptscriptstyle{(2)}}|\beta^{\scriptscriptstyle{(2)}}q)\Bigr)\right)=0 \label{eq:bilin0}\\
\sum\limits_{n\in\mathbb{Z}+1/2}\left(l_{n}^2(P,b)\cdot D^{III}_b\Bigl(\kr{F}_{c^{\scriptscriptstyle{(1)}}}(\Delta_{n}^{\scriptscriptstyle{(1)}}|\beta^{\scriptscriptstyle{(1)}}q),\kr{F}_{c^{\scriptscriptstyle{(2)}}}(\Delta_{n}^{\scriptscriptstyle{(2)}}|\beta^{\scriptscriptstyle{(2)}}q)\Bigr)\right)=0 \label{eq:bilin1}
\end{align}

We want to compare these relations with \eqref{s0},\eqref{s1}, where the central charges are equal to~1. Therefore it is natural to set $b=i$. We specify other parameters by $q=4t; P=2i\sigma$. Therefore we get
\[
Q=0, \; b^{(\eta)}=i, \; \beta^{(\eta)}=\frac14; \; P^{(\eta)}=i\sg , \; \eta=1,2, \quad \D^{\scriptscriptstyle{(1)}}_n=(\sigma+n)^2, \; \D^{\scriptscriptstyle{(2)}}_n=(\sigma-n)^2. 
\]
After this specialization we have $D^{III}_b \mapsto 2D^{III}$, $\kr{F}_{c^{\scriptscriptstyle{(1)}}}(\Delta_n^{\scriptscriptstyle{(1)}}|\beta^{\scriptscriptstyle{(1)}}q)\mapsto \kr{F}((\sigma+n)^2|t)$, $\kr{F}_{c^{\scriptscriptstyle{(2)}}}(\Delta_n^{\scriptscriptstyle{(2)}}|\beta^{\scriptscriptstyle{(2)}}q)\mapsto \kr{F}((\sigma-n)^2|t)$. In result the specialization of relation \eqref{eq:bilin0} coincides with \eqref{s0} up to coefficients. But using the recurrence relation for $\G(\sigma)$ one can prove that
\begin{align}
\frac{C(\sg+n)C(\sg-n)}{C(\sg)^2}=\frac{1}{\prod\limits_{k=1}^{2|n|-1}(k^2-4 \sg^2)^{2(2|n|-k)}(4\sg^2)^{2|n|}}=4^{-\Delta^{\NS}}(-1)^{2n}l_n(2i\sigma,i)^2,\label{C0}
\end{align}
where $2n \in \mathbb{Z}$ and the functions $l_n(2i\sigma,i)$ are specified in \eqref{l_n_irr}. Therefore the specialization of relation \eqref{eq:bilin0} coincide with \eqref{s0}.

For the specialization of relation \eqref{eq:bilin1} we  substitute $\sg\mapsto\sg+1/2$, $n \mapsto n+1/2$. Then we obtain
\[
\sum\limits_{n\in\mathbb{Z}}\left(l_{n+\frac12}^2(2i\sigma+i,i) D^{III}\Bigl(\kr{F}((\sigma+n+1)^2|t),\kr{F}((\sigma-n)^2|t)\Bigr)\right)=0 
\]
It remains to compare the coefficients. We can rewrite the coefficients in \eqref{s1}
\[
\frac{C(\sg+n+1)C(\sg-n)}{C(\sg+1/2)^2}=\frac{C((\sg+1/2)+(n+1/2))C((\sg+1/2)-(n+1/2))}{C(\sg+1/2)^2}
\]
and then using \eqref{C0} obtain $l_{n+\frac12}^2(2i\sigma+i,i)$ (with the additional factor $-4^{\Delta^{\NS}}$). This concludes the proof of \eqref{s0},\eqref{s1}. \end{proof}

\begin{Remark} Equations \eqref{eq:hatF} suggest more simple equation then \eqref{eq:eqFhat}, namely $\widehat{\kr{F}_2}=-q^{-1/2}\widehat{\kr{F}_0}$. But this second order differential equation
\begin{equation}
-q^{1/2}\sum\limits_{2n\in\mathbb{Z}}l^2_n(P,b) \kr{F}^{\scriptscriptstyle{(1)}}_{n}\kr{F}^{\scriptscriptstyle{(2)}}_{n}=\sum\limits_{2n\in\mathbb{Z}}l^2_n(P,b) D^{2}_{b,b^{-1}[\log q]}(\kr{F}^{\scriptscriptstyle{(1)}}_{n},\kr{F}^{\scriptscriptstyle{(2)}}_{n}) \label{relsh20},
\end{equation}
interchanges $\kr{F}_n$ with integer and half-integer  $n$ and does not provide analogous equation on the $\tau$~function. Therefore we need~$\widehat{\kr{F}_4}$.
\end{Remark}

\begin{Remark} Using the expressions for $H^k|1\otimes W_{\NS}\ra$ we can find the corresponding $\widehat{\kr{F}_k}$. We have for $k=1,3$ (using~\eqref{eq:FhatHirota})
\begin{align}
&\sum\limits_{2n\in\mathbb{Z}}l^2_n(P,b) D^{1}_{b,b^{-1}[\log q]}(\kr{F}^{\scriptscriptstyle{(1)}}_{n},\kr{F}^{\scriptscriptstyle{(2)}}_{n})=\widehat{\kr{F}_1}=0,\label{t1}
\\
&\sum\limits_{2n\in\mathbb{Z}}l^2_n(P,b) D^{3}_{b,b^{-1}[\log q]}(\kr{F}^{\scriptscriptstyle{(1)}}_{n},\kr{F}^{\scriptscriptstyle{(2)}}_{n})=\widehat{\kr{F}_3}=-Qq^{1/2}\kr{F}_{\NS}. 
\end{align}
In the $b^2=-1$ specialization we have trivial relations $\widehat{\kr{F}_1}=\widehat{\kr{F}_3}=0$, so we did not use these functions in the proof. But for other central charges we can write additional equations. For example, using the expression for $\widehat{\kr{F}_2}$ \eqref{eq:hatF} we get
\begin{equation}
\sum\limits_{2n\in\mathbb{Z}}l_n^2(P,b) D^{3}_{b,b^{-1}[\log q]}(\kr{F}^{\scriptscriptstyle{(1)}}_{n},\kr{F}^{\scriptscriptstyle{(2)}}_{n})=Q\sum\limits_{2n\in\mathbb{Z}}l_n^2(P,b) D^{2}_{b,b^{-1}[\log q]}(\kr{F}^{\scriptscriptstyle{(1)}}_{n},\kr{F}^{\scriptscriptstyle{(2)}}_{n}).\label{relsh32}
\end{equation}
We use these relations in Subsection \ref{ssec:calc}.
\end{Remark}

\subsection{Proof of the Painlev\'e VI $\tau$~function conjecture}
\label{ssec:P6case}

As was explained in the Introduction the initial  Gamayun--Iorgov--Lisovyy conjecture was for the Painlev\'e~VI $\tau$~function. The Painleve V and III conjectures are degenerations of that conjecture. We  prove the Painlev\'e VI conjecture below.

\begin{thm}
\label{thm:2}
The expansion of Painlev\'e VI $\tau$~function near $t=0$ can be written as 
\begin{equation}
\td{\tau}(t,\ora{\theta},s,\sg)=\sum_{n\in\mathbb{Z}} s^n C(\sg+n,\ora{\theta}) \kr{F}(\ora{\Delta},(\sg+n)^2|t) \label{taug},
\end{equation}
where 
\[\ora{\theta}=(\theta_0,\theta_t,\theta_1,\theta_\infty), \ora{\Delta}=(\Delta_0,\Delta_t,\Delta_1,\Delta_\infty),\] $\kr{F}(\ora{\Delta},(\sg+n)^2,|t)=\kr{F}_1(\ora{\Delta},(\sg+n)^2,|t)$ denotes the 4-point conformal block defined in \eqref{eq:Vir4point} for central charge $c=1$.
The coefficients $C(\sg,\ora{\theta})$ are expressed in terms of Barnes $\G$ function by the formula
\[
C(\sg,\ora{\theta})=\frac{\prod_{\epsilon,\epsilon'=\pm1}\G(1+\theta_t+\epsilon\theta_0+\epsilon'\sg) \G(1+\theta_1+\epsilon\theta_\infty+\epsilon'\sg) }{\prod_{\epsilon=\pm1}\G(1+2\epsilon\sg)}.
\]
\end{thm}

\begin{Remark}\label{rem:convergence}
It is natural to ask: do the series \eqref{eq:tau:P3} and \eqref{taug} converge? It was proven in \eqref{eq:tau:P3} that the  series \eqref{eq:tau:P3} converges absolutely and uniformly on every bounded subset of $\mathbb{C}$. One can similarly prove the convergence of \eqref{taug} on some neighborhood of $t=0$. Convergence radius of \eqref{taug} is non greater than $1$, because at the point $t=1$ conformal block have singularity. It seems to be that it really equals $1$. 

\end{Remark}

\begin{Remark}\label{rem:int const} Note that any general solution of $\tau$ form of Painlev\'e equations depends on 4 integration constants unlike $\zeta(t)$ which depend on 2 integration constants as a solution of the second order differential equation \eqref{sg3}, \eqref{sg6}. One of extra constants is a constant factor since $\tau(t)$ was defined by $t\frac{d\log \tau(t)}{dt}$. Another constant emerges from differentiation of $\zeta$ form of Painlev\'e equation. 

It is easy to see that the parameters $s$ and $\sigma$ are defined by the asymptotic behavior of $\tau(t)$ (up to discrete shift $\sigma\mapsto \sigma+1$), and are independent on constant factor of $\tau(t)$. Therefore this extra parameter, which correspond to constant factor cannot be expressed in terms of~$s,\sigma$. 

Comparing the asymptotic behavior of the corresponding $\zeta(t)$ one can see that there is no additional constant in $\zeta$ forms of the Painlev\'e equations \eqref{sg3}, \eqref{sg6}. Therefore the $\tau$~function \eqref{taug} corresponds to a solution of these equations. 
\end{Remark}

\begin{proof} The proof goes in the same way as in the Painlev\'e $\mathrm{III'_3}$ case. We substitute the conjectural expression \eqref{taug} for $\tilde{\tau}(t)$ into $D^{VI}(\tilde{\tau}(t),\tilde{\tau}(t))=0$ and collect the $s^m$--terms. It is enough to proof the vanishing of $s^0$, $s^1$ coefficients (similarly to the Painlev\'e $\mathrm{III'_3}$ case). These vanishing conditions have the form
\begin{equation}
\begin{aligned}
\sum_{n\in\mathbb{Z}}\left(C(\sg+n,\ora{\theta})C(\sg-n,\ora{\theta}) D^{VI}\Bigl(\kr{F}(\ora{\Delta},(\sigma+n)^2|t),\kr{F}(\ora{\Delta},(\sigma-n)^2|t\Bigr)\right)&=0, \\
\sum_{n\in\mathbb{Z}}\left(C(\sigma+n+1,\ora{\theta})C(\sg-n,\ora{\theta}) D^{VI}\Bigl(\kr{F}(\ora{\Delta},(\sigma+n+1)^2|t),\kr{F}(\ora{\Delta},(\sigma-n)^2|t\Bigr)\right)&=0
\end{aligned}\label{eq:s01PVI}
\end{equation}
We prove these relations using the chain vector decomposition proved in Proposition \ref{prop:Chdecomp}. Using rather cumbersome calculations (presented in  Appendix \ref{App:A}) we prove that
\begin{equation}
\sum_{2n\in\mathbb{Z}}\left(l_n^{21}(P,b|\D^{\NS}_1,\D^{\NS}_2)\cdot l_n^{34}(P,b|\D^{\NS}_3,\D^{\NS}_4) \cdot
 D^{VI}_b\Bigl(\kr{F}_{c^{\scriptscriptstyle{(1)}}}(\ora{\Delta^{\scriptscriptstyle{(1)}}},\Delta_n^{\scriptscriptstyle{(1)}}|q),\kr{F}_{c^{\scriptscriptstyle{(2)}}}(\ora{\Delta^{\scriptscriptstyle{(2)}}},\Delta_{n}^{\scriptscriptstyle{(2)}}|q)\Bigr)\right)=0, \label{eq:PVIbilin}
\end{equation}
where 
\begin{equation*}
\begin{aligned}
D^{VI}_b=& -\frac12(1-q)^3 D^{4}_{b,b^{-1}[\log q]}-(1+q)(1-q)^2 \left( q \frac{d}{dq}\right)D^{2}_{b,b^{-1}[\log q]}+
\\
&+(1-q)\left(-q(\D_2^{\NS}+\D_3^{\NS})+q(1-q)(\D_1^{\NS}+\D_4^{\NS})+\frac12\left(Q^2(1+4q+q^2)+(1-q+q^2)\right)\right)D^{2}_{b,b^{-1}[\log q]}- 
\\ 
&+\frac12q\Bigl(q(\D_2^{\NS}+\D_1^{\NS})(\D_3^{\NS}+\D_4^{\NS})-(\D_2^{\NS}-\D_1^{\NS})(\D_3^{\NS}-\D_4^{\NS})\Bigr) D^{0}_{b,b^{-1}[\log q]}+
\\
&+\frac12 q\Bigl((\D_1^{\NS}+\D_4^{\NS})(1-q)-(\D_2^{\NS}+\D_3^{\NS})(1+q)\Bigr)\left(q\frac{d}{dq}\right)D^{0}_{b,b^{-1}[\log q]}- \frac12q(1-q)\left(q\frac{d}{dq}\right)^2D^{0}_{b,b^{-1}[\log q]}. 
\end{aligned}
\end{equation*}
Here the highest weights of $\Vir^{(\eta)}, \eta=1,2$ modules are related to the highest weight of $\NSR$ module 	by the formula \eqref{eq:DeltaVirNSR}
\[
\D_{\kappa}^{\scriptscriptstyle{(1)}}=\frac{b^{-1}}{b^{-1}-b}\D_{\kappa}^{\NS}, \quad 
\D_{\kappa}^{\scriptscriptstyle{(2)}}=\frac{b}{b-b^{-1}}\D_{\kappa}^{\NS},\qquad  \kappa=1,2,3,4.
\]
Similarly to the Painlev\'e $\mathrm{III'_3}$ case the sum \eqref{eq:PVIbilin} splits into two sums, which consist on terms with integer and half-integer $n$ correspondingly. In order to get Painlev\'e~VI relations we specialize the parameters
\[
q=t, \; b=i,\quad \Rightarrow\quad   Q=0,\; c_{\NS}=c^{(\eta)}=1, \; \D_{\kappa}^{(\eta)}=\frac12 \D^{\NS}_{\kappa}, \; \eta=1,2,\; \kappa=1,2,3,4.
\]
The relation between $\Delta$ parameters in the Painlev\'e  VI equation and in equation \eqref{eq:PVIbilin} can be written in terms of parameters $\theta$ and $P$
\[
P_1=2i\theta_0,\;\;P_2=2i\theta_t,\;\;P_3=2i\theta_1,\;\;P_4=2i\theta_\infty,\quad P=2i\sigma.
\]

Then we have $D^{VI}_b \mapsto D^{VI}$ and the equation \eqref{eq:PVIbilin} reduces to \eqref{eq:s01PVI} up to the coefficients. For these coefficients we have the relations 
\[
\frac{C(\sg+n+1)C(\sg-n)}{C(\sg+1/2)^2}=\frac{C((\sg+1/2)+(n+1/2))C((\sg+1/2)-(n+1/2))}{C(\sg+1/2)^2}
\]
and
\begin{multline*}
\frac{C(\sg+n)C(\sg-n)}{C(\sg)^2}=\frac{\prod_{\epsilon=\pm}\prod\limits_{i=1-|n|}^{|n|-1}((\theta_t+\epsilon\theta_0+i)^2-\sg^2)^{|n|-i}\prod\limits_{i=1-|n|}^{|n|-1}((\theta_1+\epsilon\theta_\infty+i)^2-\sg^2)^{|n|-i}}{\prod\limits_{k=1}^{2|n|-1}(k^2-4 \sg^2)^{2(2|n|-k)}(4\sg^2)^{2|n|}}=\\=(-1)^{2n}\cdot l_n^{21}(2i\sigma,i|2\theta_0^2,2\theta_t^2) \cdot l_n^{34}(2i\sigma,i|2\theta_1^2,2\theta_\infty^2),
\end{multline*}
where $2n\in\mathbb{Z}$.
Here we used the recurrence relation on Barnes $\G$ function and explicit expressions for $l_n^{21}(P,b|\D^{\NS}_2,\D^{\NS}_1)$ \eqref{l_n}. This completes the proof. \end{proof}

\subsection{Effective algorithm for conformal blocks calculation} \label{ssec:calc}
Bilinear relations on the Virasoro conformal blocks provide efficient algorithm for calculation of the power expansions. Analogous algorithm based on Nakajima-Yoshioka relations was given in~\cite{Nakajima}.

We start with $c=1$ 4 point conformal block defined by its power expansion
\begin{equation}
\kr{F}(t)=t^{\D}\sum_{N=0}^{\infty}B(N)t^N , \quad B(0)=1. \label{eq:confblock:expans}
\end{equation}
We consider both irregular and generic cases. Substitute the expansion \eqref{eq:confblock:expans} to the \eqref{s0} in irregular case and \eqref{eq:s01PVI} in generic case. Then the equation that the $t^{\Delta^{\NS}+N}$ coefficient is 0 gives the relation which expresses $B(N)B(0)$ in terms of $l_n$ and $B(M)$ for $M<N$. Thus one can compute coefficients $B(N)$ recursively. 

This algorithm has a polynomial complexity in contrast to exponential complexity of algorithms based on the AGT expressions or on calculation of a Kac-Shapovalov matrix. Note that a Zamolodchikov recurrence formula \cite{ZCB} also provides an algorithm of a polynomial complexity. 

For instance a calculation of $B(n)$ for $n \leq 50$ by use of bilinear relations  took 33 sec in the irregular case and 256 seconds in the generic case  (we use Intel Core i3; 2.2 GHz 2 Core and program Wolfram Mathematica 8.0) . Calculation of $B(30)$ in general case took 30 sec. For comparsion a calculation of $B(30)$ using the AGT correspondence (and additional improvements such a parallelization) took 240 sec.

Such terms like $B(50)$ are important for the numerical study of the conformal blocks for $|t| \rightarrow 1$, where series \eqref{eq:confblock:expans} converges very slowly.
For instance this calculation can be used to check the formula (5.3) in \cite{ILT}. 

This method can be generalized for calculations of the Virasoro conformal blocks in case $c\neq 1$. In this case we use two bilinear relations; \eqref{t1} and \eqref{relsh32} in the irregular case and \eqref{t1g} and \eqref{relsh32g} in the generic case. In the recursion procedure the $q^{\Delta^{NS}+N}$ term contains two new terms $B^{\scriptscriptstyle{(1)}}(N)B^{\scriptscriptstyle{(2)}}(0)$ and $B^{\scriptscriptstyle{(1)}}(0)B^{\scriptscriptstyle{(2)}}(N)$. They can be found using two bilinear relations.


\section{AGT relation}
\label{sec:geom}
The AGT relation states the equality between conformal blocks in certain 2d CFT and a generating function of certain integrals on instanton moduli spaces (Nekrasov partition function). It is known \cite{BF,BMT1} that conformal field theory with $\NSR$ symmetry  corresponds by AGT to the instanton counting on the minimal resolution of $\mathbb{C}^2/\mathbb{Z}_2$. We denote this minimal resolution by $X_2$. 

Let $\mathit{M}(\mathbb{C}^2; r,N)$ be the moduli space of instantons on $\mathbb{C}^2$ with rank $r$ and  $c_2=N$. Let $\mathit{M}(X_2;r,N)$ be the moduli space of instantons on $X_2$ with rank $r$, $c_1=0$, $c_2=N$. By $Z_{\pure}(\epsilon_1,\epsilon_2,a;q)$ and $Z_{\pure}^{X_2}(\epsilon_1,\epsilon_2,a;q)$ we denote the Nekrasov instanton partition functions for pure $U(2)$ gauge theory on $\mathbb{C}^2$ and $X_2$ correspondingly. These functions are the generating functions of equivariant volumes of the corresponding instanton moduli spaces
\[
Z_{\pure}(\epsilon_1,\epsilon_2,a;q)=\sum_{N=0}^{\infty} q^N \!\!\!\!\!\!\int\limits_{\mathit{M}(\mathbb{C}^2;r,N)}1,\quad \qquad Z^{X_2}_{\pure}(\epsilon_1,\epsilon_2,a;q)=\sum_{N=0}^{\infty} q^N \!\!\!\!\!\!\int\limits_{\mathit{M}(X_2;r,N)}1,
\]

It was proven in \cite{BMT1} (see also \cite{Bruzzo:2013}) that
\begin{equation}\label{eq:Z=ZZ}
Z^{X_2}_{\pure}(\epsilon_1,\epsilon_2,a;q)=\sum_{2n \in \mathbb{Z}} \left(\frac{q^{2n^2}}{\boldsymbol{l}_n(a,\epsilon_1,\epsilon_2)} Z_{\pure}(2\epsilon_1,-\epsilon_1{+}\epsilon_2, a+2n\epsilon_1;q) Z_{\pure}(\epsilon_1{-}\epsilon_2,2\epsilon_2, a+2n\epsilon_2;q)\right),
\end{equation} 
where
\begin{equation}\label{l-vec}
\begin{aligned}   
&\boldsymbol{l}_{n}(a,\epsilon_{1},\epsilon_{2})=(-1)^{2n} s_\epsilon(2a,2n) s_\epsilon(2a+\epsilon_{1}+\epsilon_{2},2n)\\
 &s_\epsilon(x,n)=\prod_{\substack{i,j\geq 0,\;i+j < 2n\\i+j\equiv0\bmod 2}}\hspace*{-10pt}(x+i\epsilon_{1}+j\epsilon_{2}),\quad \text{for } n\geq 0\\ 
&s_\epsilon(x,n) =(-1)^ns_\epsilon(\epsilon_1+\epsilon_2-x,-n) ,\quad \text{for } n< 0.
\end{aligned}
\end{equation}
The coefficients $\boldsymbol{l}_{n}(a,\epsilon_{1},\epsilon_{2})$ are called the blow-up factors.

The AGT relation for the Virasoro algebra was proved for Virasoro algebra in \cite{Alba:2010}. In the  Whittaker limit the AGT relation states that
\[
Z_{\mathrm{\pure}}(\epsilon_1,\epsilon_2,a;q)=\left(\frac{q}{\epsilon_1^2\epsilon_2^2}\right)^{-\Delta} \kr{F}_c\left(\Delta|\frac{q}{\epsilon_1^2\epsilon_2^2}\right) ,\quad \text{where } \Delta=\frac{4a^2-(\epsilon_1+\epsilon_2)^2}{4\epsilon_1\epsilon_2}, c=1+6\frac{(\epsilon_1+\epsilon_2)^2}{\epsilon_1\epsilon_2}
\]
Then we set $\epsilon_1=b$, $\epsilon_2=b^{-1}$, $a=P$ and rewrite the right hand side of \eqref{eq:Z=ZZ} 
\[
\text{RHS}=\sum_{2n \in \mathbb{Z}} q^{2n^2-\Delta_n^{\scriptscriptstyle{(1)}}-\Delta_n^{\scriptscriptstyle{(2)}}}\frac{(-1)^{2n}2^{2\Delta_n^{\scriptscriptstyle{(1)}}+2\Delta_n^{\scriptscriptstyle{(2)}}}\beta_1^{-\Delta_n^{\scriptscriptstyle{(1)}}}\beta_2^{-\Delta_n^{\scriptscriptstyle{(2)}}}}
{s_\textrm{even}(2P|2n) s_\textrm{even}(2P+Q|2n)}
\kr{F}_{c^{\scriptscriptstyle{(1)}}}\left(\Delta_n^{\scriptscriptstyle{(1)}}\Bigr|\frac{\beta_1q}{4}\right)\kr{F}_{c^{\scriptscriptstyle{(2)}}}\left(\Delta_n^{\scriptscriptstyle{(2)}}\Bigr|\frac{\beta_2q}{4}\right),
\]
where $\beta_1, \beta_2$ are defined by \eqref{eq:beta12}, the central charges $c^{\scriptscriptstyle{(1)}}$, $c^{\scriptscriptstyle{(2)}}$ are defined by \eqref{cc} and the highest weights $\Delta_n^{\scriptscriptstyle{(1)}}, \Delta_n^{\scriptscriptstyle{(2)}}$ are defined at \eqref{momentum}. Using the equality $\Delta_n^{\scriptscriptstyle{(1)}}+\Delta_n^{\scriptscriptstyle{(2)}}=\Delta^{NS}+2n^2$ and formula \eqref{l_n_irr} we get
\[
\text{RHS}=\left(\frac{q}{4}\right)^{-\Delta^{\NS}} \sum_{2n \in \mathbb{Z}} l_n^2(P,b)
\kr{F}_{c^{\scriptscriptstyle{(1)}}}\left(\Delta_n^{\scriptscriptstyle{(1)}}\Bigr|\frac{\beta_1q}{4}\right)\kr{F}_{c^{\scriptscriptstyle{(2)}}}\left(\Delta_n^{\scriptscriptstyle{(1)}}\Bigr|\frac{\beta_2q}{4}\right).
\]
We compare the last equation with \eqref{blockdecomp} and get
\begin{equation}\label{eq:F=ZNS}
Z^{X_2}_{\pure}(b,b^{-1},P;q)=\left(\frac{q}{4}\right)^{-\Delta^{\NS}} \kr{F}_{c^{\NS}} \left(\Delta^{\NS}\Bigr|\frac{q}{4}\right).
\end{equation} This relation was proposed in \cite{BMT1} (following \cite{BF}) as the AGT relation for the $\NSR$ algebra in the Whittaker limit. This proof of \eqref{eq:F=ZNS} follows \cite{BBFLT}. This proof is based on the proof for Virasoro case, Proposition \ref{prop:Whitdecomp} and the values of $l_n$ found in Subsection \ref{ssec:melem} (which were given in \cite{BBFLT} without proof).

More general spherical (and toric) $\NSR$ conformal blocks \eqref{eq:NSRconfblock} can be expressed in terms of $\Vir$ conformal blocks by use Theorems \ref{thm:decomp}, \ref{thm:ln}. On the other hand the Bonelli, Maruyoshi, Tanzini found expression of the Nekrasov partition function on $X_2$ in terms of Nekrasov partition function on $\mathbb{C}^2$ \cite{BMT} (see eq. (2.14) in loc. cit.). But Virasoro conformal blocks are equal to Nekrasov partition functions on $\mathbb{C}^2$. Using Theorems \ref{thm:decomp}, \ref{thm:ln} (or the corresponding results in \cite{HJ}) one can see that the coefficients in the expressions in terms of conformal blocks coincide with the coefficients found in \cite{BMT}. Therefore the $\NSR$ conformal blocks coincide with the Nekrasov partition functions on $X_2$.  This finishes the proof of AGT relation for the $\NSR$ algebra in the form proposed in   \cite{BMT}.


\section{Concluding remarks}
\label{sec:disc}
\begin{itemize}
\item The conformal block bilinear equations \eqref{bilin}, \eqref{eq:PVIbilin} were proven for any central charge. It is natural to ask for the corresponding $\tau$~function equations for $c \neq 1$. Namely one can introduce $b$ dependent $\tau$~function (for irregular limit case)
\begin{equation}
\tau(b,P|q)=\sum_{n\in\mathbb{Z}}s^n C_b(P+nb) \kr{F}_{c}(\Delta(P+nb,b)|q)  \label{eq:taunb},
\end{equation}
with the coefficients $C_b(P)$ defined in terms of double Gamma function $\Gamma_2(P|b,b^{-1})$. Then we define $\tau^{\scriptscriptstyle{(1)}}(q)=\tau(b^{\scriptscriptstyle{(1)}},P_1|\beta_1 q)$, $\tau^{\scriptscriptstyle{(1)}}(q)=\tau(-(b^{\scriptscriptstyle{(2)}})^{-1},P_2|\beta_2 q)$ and ask for equation 
\[
D^{III}_b(\tau^{\scriptscriptstyle{(1)}}(q),\tau^{\scriptscriptstyle{(2)}}(q))\stackrel{?}{=}  0. 
\]
However argumentation from Subsection \ref{ssec:PIII} runs into difficulties.First note that in Subsection \ref{ssec:PIII} we used the fact that $s^m$ and $s^{m+2}$ relations are equivalent. For $c \neq 1$ this argument works only for special central charges, namely $c^{\scriptscriptstyle{(1)}}$ and $c^{\scriptscriptstyle{(2)}}$ corresponding to generalized minimal models $\mathcal{M}(1,k)$ and $\mathcal{M}(1,2-k)$, $k \in \mathbb{N}$. Another obstacle to study is the convergence of the series \eqref{eq:taunb}.

\item As was mentioned in the Introduction and Section \ref{sec:geom} equation \eqref{eq:eqFhat} has the geometrical meaning in the framework of the instanton counting on $X_2$. In would be interesting to find its the geometrical proof similarly to the Nakajima-Yoshioka proof \cite{Nakajima}. 

Probably, more fundamental question is the geometric interpretation (in terms of instanton moduli spaces) of the $\tau$~functions \eqref{eq:tau:P3}, \eqref{taug}.

\item Our approach is quite general, it seems that bilinear equations for many point conformal blocks can be obtained this way. Another possible generalization is bilinear equations on $W_N$ conformal blocks for $N>2$. It would be interesting to find bilinear equations on $\tau$~function in the corresponding isomonodromic problems.

\item Recently Litvinov, Lukyanov, Nekrasov, Zamolodchikov suggested a relation between classical conformal blocks ($c \rightarrow \infty$) and Painlev\'e $\mathrm{VI}$  \cite{LLNZ}. It would be interesting to find any relation between this fact and Gamayun--Iorgov--Lisovyy conjecture studied in our paper.
\end{itemize}


\section{Acknowledgments}
We thank  A. Belavin, B. Feigin, P. Gavrylenko, O. Lisovyy, A. Litvinov and N. Iorgov  for interest in our work and discussions. We especially thank O. Lisovyy for critical reading of the first  version of the paper and many useful comments. M.B. is grateful to the organizers of the conference ``Quantum groups and Quantum integrable systems'' Kiev, Ukraine 2013 and Vitaly Shadura for the hospitality.

The study of relation between $\Vir\oplus \Vir$ and $\mathsf{F}\oplus \NSR$ vertex operator algebras was performed under a grant funded by Russian Science Foundation (project No. 14-12-01383). The study of Painlev\'e equation was partially supported by project 01-01-14 of NASU.

\appendix

\section{Proof of relation \eqref{eq:PVIbilin}}
\label{App:A}

As we already claimed in the main text the proof of the relation \eqref{eq:PVIbilin} is similar to the proof of its irregular analogue \eqref{d3}. In this appendix we will use shorten notations 
\begin{multline*}
l_n^{21}=l_n^{21}(P,b|\D^{\NS}_2,\D^{\NS}_1),\quad l_n^{34}=l_n^{34}(P,b|\D^{\NS}_3,\D^{\NS}_4), 
\\
\kr{F}^{\scriptscriptstyle{(1)}}_{n}=\kr{F}_{c^{\scriptscriptstyle{(1)}}}(\ora{\Delta^{\scriptscriptstyle{(1)}}},\Delta_n^{\scriptscriptstyle{(1)}}|q),\quad \kr{F}^{\scriptscriptstyle{(2)}}_{n}=\kr{F}_{c^{\scriptscriptstyle{(2)}}}(\ora{\Delta^{\scriptscriptstyle{(2)}}},\Delta_{n}^{\scriptscriptstyle{(2)}}|q).
\end{multline*}
Similarly to \ref{ssec:PIII} we define the functions $\widehat{\kr{F}_k}$ by the formulae
\[
\sum\limits_{k=0}^{\infty}\widehat{\kr{F}_k} \frac{\alpha^k}{k!}=
{}_{34}\la 1\otimes W_{\NS}|e^{\alpha H}|1\otimes W_{\NS}\ra_{21}=\sum\limits_{k=0}^{\infty}{}_{34}\la 1\otimes W_{\NS}|H^{k}|1\otimes W_{\NS}\ra_{21} \frac{\alpha^k}{k!},
\]
where the operator $H$ was defined in \eqref{eq:Hdefin}. As in irregular case we have 
\begin{equation}\label{eq:FhatPVIHirota}
\widehat{\kr{F}_k}=\sum\limits_{2n\in\mathbb{Z}}\left(l_n^{21}\cdot l_n^{34} \cdot D^{k}_{b,b^{-1}[\log q]}(\kr{F}^{\scriptscriptstyle{(1)}}_{n},\kr{F}^{\scriptscriptstyle{(2)}}_{n})\right),
\end{equation}
where we used decomposition \eqref{Chdecomp}. On the other side we can use expression \eqref{eq:HNSR} of the operator $H$ in terms of $\mathsf{F}\oplus\NSR$ generators and calculate
\begin{align*}
H|1\otimes W_{\NS}\ra_{21}=&-\sum_{r\in\mathbb{Z}_{\geq 0}+\frac12} q^{r/2} f_{-r}|1\ra\otimes |\td{W_{\NS}}\ra_{21}
\\
H^2|1\otimes W_{\NS}\ra_{21}=&-Q\sum_{r\in\mathbb{Z}_{\geq 0}+\frac12} 2r q^{r/2} f_{-r}|1\ra\otimes|\td{W_{\NS}}\ra_{21}-
\sum_{r\in\mathbb{Z}_{\geq 0}+\frac12}q^{r/2}|1\ra\otimes G_{-r}| \td{W_{\NS}}\ra_{21}-\\
&-
\sum_{r,s\in\mathbb{Z}_{\geq 0}+\frac12,s\neq r}2r\D_2^{\NS}q^{\frac{r+s}{2}}f_{-r}f_{-s}|1\ra\otimes |W_{\NS}\ra_{21}.
\end{align*}
For $H^3$ and $H^4$ we omit terms which do not contribute to $\widehat{\kr{F}_k}$, $k\leq 4$
\begin{multline*}
H^3|1\otimes W_{\NS}\ra_{21}=
-Q^2\sum_{r\in\mathbb{Z}_{\geq 0}+\frac12}(2r)^2 q^{r/2}f_{-r}|1\ra\otimes|\td{W_{\NS}}\ra_{21}-Q \sum_{r\in\mathbb{Z}_{\geq 0}+\frac12}2r q^{r/2}|1\ra\otimes G_{-r}|\td{W_{\NS}}\ra_{21}+\\+
\!\!\!\sum_{r,s\in\mathbb{Z}_{\geq 0}+\frac12}\!\!\!q^{r/2}f_{-s}|1\ra\otimes G_{s}G_{-r}|\td{W_{\NS}}\ra_{21}-\!\!\!\!\!\!\sum_{r,s\in\mathbb{Z}_{\geq 0}+\frac12,r\neq s}\!\!\!\!\!\!2r\D_2^{\NS}q^{\frac{r+s}{2}}(f_{-r}G_{-s}-f_{-s}G_{-r})|1\ra\otimes|W_{\NS}\ra_{21}+\ldots
\end{multline*}
\begin{multline*}
H^4|1\otimes W_{\NS}\ra_{21}=-Q^2\!\!\!\!\!\sum_{r\in\mathbb{Z}_{\geq 0}+\frac12}(2r)^2 q^{r/2}|1\ra\otimes G_{-r} |\td{W_{\NS}}\ra_{21}+
\\+
\!\!\!\!\!\!\sum_{r,s\in\mathbb{Z}_{\geq 0}+\frac12}\!\!\!q^{r/2}|1\ra\otimes G_{-s}\left(2L_{s-r}+\frac{c_{\NS}}2(r^2-\frac14)\delta_{r,s}\right) |\td{W_{\NS}}\ra_{21}
-\sum_{r,s\in\mathbb{Z}_{\geq 0}+\frac12}q^{r/2}|1\ra\otimes G_{-s}G_{-r} G_s|\td{W_{\NS}}\ra_{21}+
\\
+\sum_{r,s\in\mathbb{Z}_{\geq 0}+\frac12,r\neq s} \!\!\! 2r\D_2^{\NS} q^{\frac{r+s}{2}}|1\ra\otimes \Bigl(G_{-s}G_{-r}-G_{-r}G_{-s}\Bigr)|W_{\NS}\ra_{21}+\ldots 
\end{multline*}
Then we multiply these equations by ${}_{34}\la W_{\NS}|$ and calculate the corresponding $\widehat{\kr{F}_k}$. For $\widehat{\kr{F}_0}$ and $\widehat{\kr{F}_2}$ one can easily~get
\begin{equation}
\widehat{\kr{F}_0}=\kr{F}_{\NS},  \qquad\widehat{\kr{F}_2}=-\frac{q^{1/2}}{1-q}\td{\kr{F}_{\NS}},\label{t2g}
\end{equation}
where $\kr{F}_{\NS}=\kr{F}_{c_{\NS}}(\overrightarrow{\D^{\NS}},\D^{\NS}|q)$, and similarly for $\td{\kr{F}}_{\NS}$, see \eqref{eq:FNS:4p}. For the calculation of $\widehat{\kr{F}_4}$ we use \eqref{eq:W:NS} and matrix elements 
\[\la W_{\NS}|L_0|W_{\NS}\ra=\left(q\frac{d}{dq}\right)\kr{F}_{\NS},\quad \la W_{\NS}|L_0^2|W_{\NS}\ra=\left(q\frac{d}{dq}\right)^2\kr{F}_{\NS},\quad\la \td{W_{\NS}}|L_{0}|\td{W_{\NS}}\ra=q\frac{d}{dq}\td{\kr{F}}_{\NS}\] 
and get:
\begin{gather*}
\widehat{\kr{F}_4}=-Q^2  \sum_{r\in\mathbb{Z}_{\geq 0}+\frac12}(2r)^2 q^r \td{\kr{F}}_{\NS}
+2\left(\sum_{r,s\in\mathbb{Z}_{\geq 0}+\frac12, r<s}\!\!\!\!\!\!q^{s}\left((s-r)(\D_2^{\NS}+\frac12)-\D_1^{\NS}+q\frac{d}{dq}\right)\td{\kr{F}}_{\NS}+\right.
\notag \\
\left.\!\!\!\!\!\!\sum_{r,s\in\mathbb{Z}_{\geq 0}+\frac12, s<r}\!\!\!\!\!\!q^{r}\left((r-s)(\D_3^{\NS}+\frac12)-\D_4^{\NS}+q\frac{d}{dq}\right)\td{\kr{F}}_{\NS}+\sum_{r\in\mathbb{Z}_{\geq 0}+\frac12}q^r q\frac{d}{dq}\td{\kr{F}}_{\NS}\right)+\frac{c_{\NS}}2\!\!\!\sum_{r\in\mathbb{Z}_{\geq 0}+\frac12}q^{r}(r^2-\frac14)\td{\kr{F}}_{\NS}
\ny
\\
-\sum_{r,s\in\mathbb{Z}_{\geq 0}+\frac12}q^{r+s}\left(2r\D_3^{\NS}-\D_4^{\NS}+q\frac{d}{dq}\right)\left(2s\D_2^{\NS}-\D_1^{\NS}+q\frac{d}{dq}\right)\kr{F}_{\NS}- \!\!\!\!\!\!\sum_{r,s\in\mathbb{Z}_{\geq 0}+\frac12,r\neq s}\!\!\!\!\!\!4q^{r+s}r(s-r) \D_2^{\NS} \D_3^{\NS}\kr{F}_{\NS}
\end{gather*}
Substituting $\td{\kr{F}}_{\NS}$ and $\kr{F}_{\NS}$ from \eqref{t2g} and performing the summation we obtain
\begin{gather*}
\widehat{\kr{F}_4}=Q^2 \frac{1+6q+q^2}{(1-q)^2}\widehat{\kr{F}_2}- \frac{q(-\D_1^{\NS}+\D_2^{\NS}+q(\D_1^{\NS}+\D_2^{\NS}))(-\D_4^{\NS}+\D_3^{\NS}+q(\D_3^{\NS}+\D_4^{\NS})) }{(1-q)^4}\widehat{\kr{F}_0}-\ny
\\
-\frac{q(-\D_1^{\NS}-\D_4^{\NS}+\D_2^{\NS}+\D_3^{\NS}+(\D_1^{\NS}+\D_2^{\NS}+\D_3^{\NS}+\D_4^{\NS})q)}{(1-q)^3}q\frac{d}{dq}\widehat{\kr{F}_0}-\frac{q}{(1-q)^2}\left(q\frac{d}{dq}\right)^2\widehat{\kr{F}_0}-\ny
\\
-2\frac{q(\D_2^{\NS}+\D_3^{\NS}+1-(\D_1^{\NS}+\D_4^{\NS})(1-q))}{(1-q)^2}\widehat{\kr{F}_2}-2 \frac{1+q}{1-q} q \frac{d}{dq}\widehat{\kr{F}_2}+\ny
\\
+ \frac{(1+q)^2}{(1-q)^2}\widehat{\kr{F}_2}-(1+2Q^2)\frac{q}{(1-q)^2}\widehat{\kr{F}_2}+
4\D_2^{\NS}\D_3^{\NS}\frac{q^2}{(1-q)^4}\widehat{\kr{F}_0}
\end{gather*}

Using \eqref{eq:FhatPVIHirota} we finally get \eqref{eq:PVIbilin}.

\begin{Remark}
As in the Painlev\'e $\mathrm{III'}_3$ case we can calculate $\widehat{\kr{F}_k}$ for $k=1,3$
\begin{align}
&\widehat{\kr{F}_1}= \sum\limits_{2n\in\mathbb{Z}}l_n^{21} l_n^{34} D^{1}_{b,b^{-1}[\log q]}(\kr{F}^{\scriptscriptstyle{(1)}}_{n},\kr{F}^{\scriptscriptstyle{(2)}}_{n})=0 \label{t1g}\\
&\widehat{\kr{F}_3} = \sum\limits_{2n\in\mathbb{Z}}l_n^{21} l_n^{34} D^{3}_{b,b^{-1}[\log q]}(\kr{F}^{\scriptscriptstyle{(1)}}_{n},\kr{F}^{\scriptscriptstyle{(2)}}_{n})=-Q\frac{q^{1/2}(1+q)}{(1-q)^2}\td{\kr{F}_{\NS}} \label{t3g}
\end{align}
These relations are trivial in case $b^2=-1$ due to $l_{n}^{21}l_{n}^{34}=l_{-n}^{21}l_{-n}^{34}$.
For $b^2\neq -1$ we can use \eqref{t2g} and get
\begin{equation}
(1-q)\sum\limits_{2n\in\mathbb{Z}}l_n^{21} l_n^{34} D^{3}_{b,b^{-1}[\log q]}(\kr{F}^{\scriptscriptstyle{(1)}}_{n},\kr{F}^{\scriptscriptstyle{(2)}}_{n})=Q(1+q)\sum\limits_{2n\in\mathbb{Z}}l_n^{21} l_n^{34} D^{2}_{b,b^{-1}[\log q]}(\kr{F}^{\scriptscriptstyle{(1)}}_{n},\kr{F}^{\scriptscriptstyle{(2)}}_{n})\label{relsh32g}.
\end{equation}
\end{Remark}

\noindent \textsc{Landau Institute for Theoretical Physics, Chernogolovka, Russia,\\
Institute for Information Transmission Problems, Moscow, Russia,\\
National Research University Higher School of Economics, International Laboratory of Representation Theory and Mathematical Physics,\\
Independent University of Moscow, Moscow, Russia}

\emph{E-mail}:\,\,\textbf{mbersht@gmail.com}\\

\noindent
\textsc{Department of Physics, Taras Shevchenko National University of Kiev, Kiev, Ukraine\\
Bogolyubov Institute for Theoretical Physics, Kiev,  Ukraine\\
National Research University Higher School of Economics, Moscow, Russia.}

\emph{E-mail}:\,\,\textbf{shch145@gmail.com}

\end{document}